\documentclass[letterpaper, 11pt]{article}

\usepackage{amsmath, amssymb, amsfonts, amsthm}
\usepackage{bm}
\usepackage{mathtools}
\usepackage{mathrsfs}
\usepackage{geometry}
\usepackage{setspace}
\usepackage[utf8x]{inputenc}
\usepackage{algpseudocode,algorithmicx,algorithm}
\usepackage[T1]{fontenc}
\usepackage{lmodern}
\usepackage{graphicx}
\usepackage{enumerate}
\usepackage[natural, dvipsnames, pdftex]{xcolor}
\usepackage{tikz}
\usepackage{verbatim}
\usepackage{hyperref}
\usepackage{cancel}
\usepackage[all]{xy}
\usepackage{mdframed}
\usepackage[capitalize]{cleveref}
\usepackage{todonotes}
\usepackage{algpseudocode}
\usepackage{comment}

\usepackage[numbers]{natbib}

\usepackage{enumitem}

\usepackage{indentfirst}
\theoremstyle{plain}
\newtheorem{thm}{Theorem}[section]
\newtheorem{lemma}[thm]{Lemma}

\newtheorem{cor}[thm]{Corollary}

\theoremstyle{definition}

\theoremstyle{remark}
\newtheorem{rmk}[thm]{Remark}

\newcommand{\F}{\mathbb{F}}

\newcommand{\bc}{\mathbf{c}}

\newcommand{\cC}{\mathcal{C}}

\newcommand{\supp}{\mathrm{supp}}

\newcommand{\E}{\mathbb{E}}

\newcommand{\wt}{\mathrm{wt}}
\newcommand{\rank}{\mathrm{rank}}
\newcommand{\Vol}{\mathrm{Vol}}

\begin{document}

\title{Improvement of the Gilbert-Varshamov Bound for Linear Codes and Quantum Codes
\thanks{
C. Yuan is with School of Computer Science, Shanghai Jiao Tong University. (Email: \href{chen_yuan@sjtu.edu.cn}{chen\_yuan@sjtu.edu.cn}) R. Zhu is with School of Computer , Shanghai Jiao Tong University. (Email: \href{sjtuzrq7777@sjtu.edu.cn}{sjtuzrq7777@sjtu.edu.cn})
}}

\author{Chen Yuan, Ruiqi Zhu}
\maketitle

\begin{abstract}
The \emph{Gilbert--Varshamov} (GV) bound is a central benchmark in coding theory, establishing existential guarantees for error-correcting codes and serving as a baseline for both Hamming and quantum fault-tolerant information processing. Despite decades of effort, improving the GV bound is notoriously difficult, and known improvements often rely on technically heavy arguments and do not extend naturally to the quantum setting due to additional self-orthogonality constraints.

In this work we develop a concise probabilistic method that yields an improvement over the classical GV bound for $q$-ary linear codes. For relative distance $\delta=d/n<1-1/q$, we show that an $[n,k,d]_q$ linear code exists
whenever
\[
\frac{q^{k}-1}{q-1}\;<\;\frac{c_\delta \sqrt{n}\, q^{n}}{\Vol_q(n,d-1)},
\]
for positive constant $c_\delta$ depending only on $\delta$, where $\Vol_q(n,d-1)$ denotes the volume of a $q$-ary Hamming ball.

We further adapt this approach to the quantum setting by analyzing symplectic self-orthogonal structures. For $\delta<1-1/q^2$, we obtain an improved quantum GV bound: there exists a
$q$-ary quantum code $[[n,\,n-k,\,d]]$ provided that
\[
\frac{q^{2n-k}-1}{q-1}
\;<\;
\frac{c_\delta \sqrt{n}\cdot q^{2n}}{\sum_{i=0}^{d-1}\binom{n}{i}(q^2-1)^i}.
\]
In particular, our result improves the standard quantum GV bound by an $\Omega(\sqrt{n})$ multiplicative factor.
\end{abstract}


\section{Introduction}
Error-correcting codes play central role in ensuring reliable data transmission and storage in the presence of noise. One of the foundational results in coding theory is the \emph{Gilbert-Varshamov} bound (GV bound), which establishes the existence of codes with asymptotically trade-offs between rate and minimum distance.

Let $A_q(n,d)$ denote the maximum number of codeswords in a code of length $n$ and minimum distance $d$ over $\F_q$. The GV bound, which asserts that
\begin{equation}
    A_q(n,d)\geq \frac{q^n}{\Vol_q(n,d-1)}
\label{eq:GV bound}
\end{equation}
where
\[
\Vol_q(n,d)=\sum_{i=0}^d \binom{n}{i}(q-1)^i
\]
is one of the most well-known and fundamental results in coding theory. This inequality was first established by Gilbert~\cite{gilbert1952comparison} in 1952. Subsequently, Varshamov~\cite{varshamov1957estimate} considered the case of linear code and, on this basis, proved that there exists $[n,k,d]_q$-linear code if
\begin{equation}
    q^k<\frac{q^n}{\Vol_q(n,d-1)}.
\end{equation}
This bound is used extensively in the coding theory literature~\cite{macwilliams1977theory,pless1998handbook}, and has been generalized to numerous contexts~\cite{gu1993generalized,marcus2002improved,tolhuizen2002generalized}. In addition, the Gilbert–Varshamov bound has connections to fast-encodable codes and local error-correction codes~\cite{brehm2025linear,gopi2018locally}. In the past few decades, many researchers have devoted considerable effort to constructing families of codes that meet the Gilbert–Varshamov bound and have obtained a wealth of results~\cite{sendrier2004linear,shum2002low,kasami1974gilbert}.

Improving upon the \emph{Gilbert-Varshamov} bound is a notoriously difficult task~\cite{pless1998handbook}, \cite{tsfasman2013algebraic}. Over the decades, researchers have sought to improve, generalize, or tighten the GV bound in various directions. The breakthrough work of Tsfasman~\cite{1982Modular} led to an asymptotic improvements of \eqref{eq:GV bound}, but only for alphabets of size $q\geq 49$(see also the paper \cite{xing2003nonlinear}). In 2004, Jiang and Vardy~\cite{jiang2004asymptotic} improved the GV bound for nonlinear binary codes to
\[
A_2(n,d)\geq cn\frac{2^n}{\Vol_2(n,d-1)}
\]
for $d/n\leq 0.499$, where the constant $c$ depends only on the ratio $d/n$. Following the breakthrough on binary non-linear codes, Vu and Wu \cite{vu2005improving} extended the \emph{Gilbert-Varshamov} bound for $q$-ary codes, whereas Gaborit and Zemor \cite{gaborit2008asymptotic} showed that certain asymptotic families of linear binary $[n,\frac{n}{2}]$ random double circulant codes can achieve the same improved \emph{Gilbert-Varshamov} bound. Recently, by analyzing the moments of certain functionals associated with the code, Hao et, al. \cite{hao2022distribution} proved the \emph{Gilbert-Varshamov} bound for linear codes over $\F_q$ to
\[
q^k<cn^{1/2}\frac{q^n}{\Vol_q(n,d-1)}
\]
where $c>0$ may only depend on $d/n$ and the code rate $\frac{k}{n}$ can be taken for any constant value. We find that their argument is quite involved which relies on bounding the moment of random variables and the convergence of a Gumbel distribution. It is not clear if their method can be generalized straightforward to other metric.    

Quantum error-correcting codes play a central role in protecting quantum states against decoherence and noise, and are widely regarded as indispensable for reliable quantum information processing and communication~\cite{shor1995scheme,cao2021mds,leverrier2022quantum,knill1998resilient}. In the past few years, the theory of quantum error-correcting codes has been developed rapidly, attracting substantial attention from the coding theory community~\cite{kitaev2003fault,terhal2015quantum,bombin2006topological,panteleev2022asymptotically,movassagh2020constructing}. Various constructions are given through classical coding. 
However, it is still a great challenge to construct good quantum codes. Shor~\cite{shor1995scheme} gave the first construction of quantum codes. Subsequently, Steane~\cite{steane1996error} proposed that quantum codes can be constructed from classical self-orthogonal linear codes. Following these results, constructions of quantum codes through classical codes with certain self-orthogonality have been extensively studied and investigated. For instance, quantum codes can be obtained from Euclidean, Hermitian self-orthogonal, or Symplectic self-orthogonal codes~\cite{aly2007quantum,bierbrauer2000quantum,grassl1999quantum,ketkar2006nonbinary}.

While the classical \emph{Gilbert-Varshamov} bound serves as a fundamental benchmark for classical error-correcting codes, the pursuit of analogous results in the quantum setting has attracted considerable attention. Unlike the classical case, quantum codes must satisfy additional self-orthogonality constraints, making the direct extension of the GV bound to the quantum domain a highly nontrivial task. Over the past few years, the theory of quantum error-correcting codes has developed rapidly, and researchers have devoted substantial effort to constructing good quantum codes, in which the GV bound continues to play a crucial role \cite{ouyang2014concatenated,niehage2007nonbinary,feng2006asymptotic,guan2022construction}.

The first quantum GV bound was established by Calderbank and Shor as well as Steane~\cite{2002Quantum} for the binary case, marking the quantum analogue of the classical GV bound. Following this pioneering result, Ashikhmin and Knill~\cite{ashikhmin2002nonbinary} generalized the binary quantum GV bound to the $q$-ary case by employing the framework of stabilizer codes. Later, Feng and Ma~\cite{feng2004finite} derived a finite version of the GV bound for classical Hermitian self-orthogonal codes and subsequently applied it to quantum codes, leading to a tighter finite-length quantum GV bound. In 2011, Jin and Xing~\cite{jin2011quantum} established a GV-type bound for symplectic self-orthogonal codes and, based on this construction, obtained an Gilbert-Varshamov type bound for quantum codes. 

Despite these advancements, the quantum GV bound remains far from being tight. 
One interesting question is that whether we can improve this bound like the Hamming metric case.
Notably, the approach of~\cite{jiang2004asymptotic} fails to be generalized to quantum code as the quantum code is mostly obtained from linear codes in Hamming metric. Moreover, if we construct quantum code from classic Hamming code, we require that this Hamming code should be an orthogonal code. This also makes the approach of~\cite{gaborit2008asymptotic} inapplicable. We also note that the approach in~\cite{hao2022distribution} does not extend naturally to the quantum setting. This is because their method crucially relies on treating codewords of a random linear code as essentially independent, thereby approximating the minimum (Hamming) weight by the minimum of i.i.d. binomial random variables. 
In contrast, quantum codes are subject to additional orthogonality constraints, which induce inherent dependencies among the corresponding vectors and preclude such an independence-based approximation. Moreover, the minimum distance of quantum code is exactly the dual distance of corresponding orthogonal code which makes the argument more involved. Therefore, further improvement requires new probabilistic or algebraic approaches that can handle the inherent self-orthogonality constraints of quantum codes.


\subsection{Our Results}
In this work, we present a general framework for improving the \emph{Gilbert-Varshamov} bound which yields the improved GV bounds for both Hamming distance codes and the quantum codes.

\paragraph{Improvement of Hamming GV Bound.}

For Hamming distance, we first apply the Chernoff bound to show that when $\ell$ vectors $\vec{v}_1, \ldots, \vec{v}_\ell$ are randomly chosen from the Hamming ball $B_q(n,\delta n)$, the probability that their sum $\sum_{i=1}^{\ell} \vec{v}_i$ also lies in $B_q(n,\delta n)$ is at most $2^{-h_{\delta}n}$. In particular, we prove the following lemma.
\begin{lemma}\label{lem:1}
    Assume $d-1= \delta n$ with some constant $\delta\in (0,1-\frac{1}{q})$. Let $\vec{v}_1,\ldots \vec{v}_\ell$ be $n\geq \ell \geq 2$ random vectors distributed uniformly at random in the Hamming ball $B_q(n,\delta n)\subseteq \F_q^n$. Then, the probability that $\sum_{i=1}^{\ell}\vec{v}_i\in B_q(n,\delta n)$ is at most $2^{-h_{\delta}n}$ for some constant $h_{\delta}\in (0,1)$.
\end{lemma}
Given a random linear code $\mathcal{C}$ and its generator matrix $G$, we can obtain the codeword  $\vec{c}_m:=\vec{m}^TG$. Note that the weight of $\vec{c}_m$ is the same as the weight of $\vec{c}_{\lambda m}$. Thus, it suffices to consider a subset $W=(\F_q^k-\{0\})/\F^*_q$ of size $\frac{q^k-1}{q-1}$ and analyze the minimum weight of $\vec{c}_m$ for $\vec{m}\in W$. Let $E_X$ denote the event that all random vectors $\vec{c}_m$ with $\vec{m}\in X$ have weight less than $d$, and we apply the inclusion–exclusion principle to expand $\Pr[\bigcup_{\vec{m}\in W}E_{\vec{m}}]$ up to the third term:
\[
    \Pr[\bigcup_{\vec{m}\in W}E_{\vec{m}}]\leq \sum_{\vec{m}\in W}\Pr[E_{\vec{m}}]-\sum_{X\in \binom{W}{2}}\Pr[E_X]+\sum_{X\in \binom{W}{3}}\Pr[E_X].
\]
If the vectors in $X$ are linearly independent, we can treat them as i.i.d from $B_q(n,\delta n)$. In contrast, if they are linearly dependent, the upper bound can be derived by applying Lemma~\ref{lem:1}. Therefore, we obtain:
\[
\Pr[\bigcup_{\vec{m}\in W}E_{\vec{m}}]\leq b-\frac{b^2}{2}(1-2^{-h_\delta n-o(n)}+\frac{b^3}{6})
\]
where $b=\frac{\Vol_q(n,\delta n)}{q^n}|W|$. It suffices to make the right-hand side of the inequality smaller than $1$, which leads to a warm-up theorem.
\begin{thm}
    Let $d=\delta n$ for some $\delta \in(0,1-\frac{1}{q})$. There exists $[n,k,d]$-linear code over $\F_q$ if 
    \[
    \frac{q^k-1}{q-1}<\frac{1.7q^n}{\sum_{i=0}^{d-1}\binom{n}{i}(q-1)^i}.
    \]
\end{thm}
We can refine our argument by applying the inclusion-exclusion principle up to the $i$-th term.
\begin{equation}
    \Pr[\bigcup_{\vec{m} \in W} E_{\vec{m}} ]
    \leq \sum_{i=1}^t (-1)^{i-1} \sum_{X \in \binom{W}{i}} \Pr[E_X].
\end{equation}
for some odd integer $t = \sqrt{h_{\delta} \log_q 2\cdot \frac{n}{2}}$ where $h_{\delta}$ is given in the above lemma. Using a similar upper bound on $E_X$ as in the warm-up theorem, we show that for any linearly independent vectors $Y := \{\vec{x}_1, \ldots, \vec{x}_s\}$, we have the followings:
\[
\sum_{i=s}^t \sum_{\substack{Y \subseteq X \subseteq span\{Y\}, \\ |X| = i}} \Pr[E_X] 
\leq a^s (1 + 2^{-h_{\delta}\cdot \frac{n}{2}})
\]
where $a=\frac{\Vol_q(n,d-1)}{q^{n}}$.

Combining the above analysis, we obtain our main theorem for the improved GV bound of Hamming distance.
\begin{thm}
    Let $d=\delta n$ for some $\delta\in (0,1-\frac{1}{q})$. There exists $[n,k,d]$-linear code over $\F_q$ if
    \[
    \frac{q^k-1}{q-1}<\frac{c_\delta\sqrt{n}\cdot q^n}{\sum_{i=0}^{d-1}\binom{n}{i}(q-1)^i}
    \]
    for some constant $c_{\delta}>0$.
\end{thm}

\paragraph{Improvement of Quantum GV Bound.}
The analysis of the improved Quantum GV bound is more complicated. The reason is that most of quantum codes constructions require that each codeword should be orthogonal to the others. Thus, we focus on one family of quantum code which using symplectic distance. Let  $B_q^S(2n,d)\subseteq \F_q^{2n}$ be the symplectic ball of radius $d$. Given a vector $\vec{u}\in \F_q^{2n}$, we use $\langle\vec{u}\rangle$ to denote the space spanned by $\vec{u}$ and $\langle\vec{u}\rangle^{\perp_S}$ be the orthogonal space of $\vec{u}$ under symplectic metric. Given a codeword $\vec{u}$ in some random quantum code $\mathcal{C}$, by the orthogonal property of quantum code, the rest codeword  in $\mathcal{C}$ must lie in $\langle\vec{u}_i\rangle^{\perp_S}$. If  $\left|B_q^S(2n,d-1)\bigcap_{i}\langle\vec{u}_i\rangle^{\perp_S}\right|$ has large volume, then we can not apply the same argument to quantum code to derive the improved GV bound.
Thus, our first step to prove the following theorem.
\begin{thm}\label{thm:4}
    Let $\ell=\frac{\sqrt{a_{\delta}n}}{2}$ where $a_{\delta}$ is an constant, $\vec{u}_1,\ldots,\vec{u}_\ell$ be linearly independent random vectors in $B_q^S(2n,d-1)$. Then, it holds
    \[
    \frac{1-q^{-\Omega(n^{2/3})}}{q^\ell}\Vol_q^S(2n,d-1)\leq \Bigl|B_q^S(2n,d-1)\cap\langle\vec{u}_1\rangle^{\perp_S}\cap\ldots\cap\langle\vec{u}_\ell\rangle^{\perp_S}\Bigr|\leq\frac{1+q^{-\Omega(n^{2/3})}}{q^\ell}\Vol_q^S(2n,d-1).
    \]
\end{thm}

Extending Lemma~\ref{lem:1} from the Hamming ball to the symplectic ball is natural. We first prove that with high probability,  among all sets of $\ell$ linearly independent vectors, the mutually orthogonal ones are contained in the largest number of symplectic self-orthogonal dual codes $\cC^{\perp_S}$. When we consider the case where these $\ell$ vectors are mutually orthogonal, \cref{thm:4} implies that the set of “bad” $\ell$-vectors still constitutes only an exponentially small fraction of all mutually orthogonal $\ell$-vectors.

\begin{thm}
    Assume $d-1=\delta n$ with some constant $\delta\in (0,1-\frac{1}{q^2})$. Let $\ell=\min\{\frac{\sqrt{a_{\delta} n}}{2},\sqrt{h_{\delta}\log_q2\cdot \frac{n}{2}}\}$, $\vec{v}_1,\ldots,\vec{v}_{\ell}$ be $n\geq \ell\geq 2$ random vectors distributed uniformly at random in the symplectic ball $B_q^S(2n,\delta n)\subseteq \F_q^{2n}$ and they are mutually symplectic orthogonal. Then, the probability that $\sum_{i=1}^{\ell}\vec{v}_i\in B_q^S(2n,\delta n)$ is at most $2^{-h'_{\delta}n}$ for some constant $h'_\delta\in (0,1)$.
\end{thm}

Hence, by combining the above results and applying the same  method as in the Hamming case, we finally obtain the improved quantum \emph{Gilbert-Varshamov} bound.
\begin{thm}\label{thm:intro-quantum-GV-bound}
    Let $d=\delta n$ for some $\delta\in (0,1-\frac{1}{q^2})$. There exists $[2n,2n-k,d]$-linear symplectic self-orthogonal dual code $\cC^{\perp_S}$ over $\F_q$ if
    \[
    \frac{q^{2n-k}-1}{q-1}<\frac{c_\delta \sqrt{n}\cdot q^{2n}}{\sum_{i=0}^{d-1}\binom{n}{i}(q^2-1)^i}
    \]
    for some constant $c_{\delta}>0$.
\end{thm}

According to the connection between symplectic self-orthogonal dual codes over $\F_q$ and $q$-ary quantum codes (see Lemma~\ref{lem:construct-quantum-codes}), \cref{thm:intro-quantum-GV-bound} directly implies the the existence of quantum code.

\begin{cor}
    Let $1\leq k\leq n$, $d=\delta n$ for some $\delta\in (0,1-\frac{1}{q^2})$. There exists a $q$-ary $[[n,n-k,d]]$ quantum code $\cC$ if
    \[
    \frac{q^{2n-k}-1}{q-1}<\frac{c_\delta \sqrt{n}\cdot q^{2n}}{\sum_{i=0}^{d-1}\binom{n}{i}(q^2-1)^i}
    \]
    for some constant $c_{\delta}>0$.
\end{cor}

\subsection{Organizations}
We first recall the basic definitions and results concerning classical codes and symplectic self-orthogonal codes, as well as some preliminary probabilistic tools in Section 2. In Section 3, we prove a simple improvement of the classical GV bound for linear codes based on Bonferroni inequalities. 
In Section 4 we increase the number of terms in the Bonferroni inequalities to the order of $\Omega(\sqrt{n})$. By applying certain relaxation and bounding techniques, we eventually arrive at the $\Omega(\sqrt{n})$ improvement of the GV bound. In Section 5, we manage to prove the same improvement holding for quantum GV bound as well.

\section{Preliminaries}
Let $q$ be a positive power and $\F_q$ be the corresponding finite field. Given positive integers $n,k$, a $[n,k]_q$ linear code $\cC$  is a subspace with dimension $k$ of $\F_q^n$. For a (not necessarily linear) code of length $n$ over $\F_q$ with $M$ codewords, the rate of code $R$ is defined to be $\log_q(M)/n$, which is $k/n$ for linear codes. The \emph{Hamming} distance between any two codewords $\vec{x}=(x_1,\ldots,x_n)$ and $\vec{y}=(y_1,\ldots,y_n)$ in $\F_q^n$ is defined as
\[
d(\vec{x},\vec{y}):=|\{1\leq i\leq n,x_i\neq y_i\}|,
\]
and the \emph{Hamming} distance of a code $\cC$ is $d(\cC)=\min_{\substack{\vec{x}, \vec{y} \in C \\ \vec{x} \ne \vec{y}}} d(\vec{x}, \vec{y})$. The \emph{Hamming} weight of codeword $\vec{x}$ be defined as $\wt(\vec{x})=d(\vec{x},\vec{0})$. For linear codes $\cC$, $d(\cC)$ is equal to the minimal weight over all nonzero codewords $min_{\vec{0}\neq \vec{x}\in \cC}\wt(\vec{x})$. 
A linear code $\cC$ is said to be a $[n,k,d]$-linear code if $\cC$ has code length $n$, dimension $k$ and minimum distance $d$.
Let $B_q(n,d)=\{\vec{x}\in \F_q^n:\wt(\vec{x})\leq d\}$ denote as the \emph{Hamming} ball of radius $d$ in $\F_q^n$ and $\Vol_q(n,d)$ denote the volume of the ball. It is well-known that a code with a minimum distance $d$ can correct up to $\lfloor \frac{d-1}{2}\rfloor$ errors. The relative minimum distance $\delta$ is defined as the ratio $d/n$. In coding theory, the trade-off between the code rate $R$ and error-correcting ability $\delta$ is a central research topic.

We next introduce the symplectic structure and the notion of symplectic self-orthogonal codes, which play a fundamental role in the construction and analysis of quantum error-correcting codes. For vectors $\vec{a}=(a_1,\ldots,a_n), \vec{b}=(b_1,\ldots,b_n), \vec{a}'=(a_1',\ldots,a_n'), \vec{b}'=(b_1',\ldots,b_n')\in \F_q^{n}$, the symplectic inner product $\langle,\rangle_S$ is defined by
\[
\langle (\vec{a}\mid \vec{b}),(\vec{a}'\mid\vec{b}')\rangle_S=\langle \vec{a},\vec{b}'\rangle_E-\langle\vec{b},\vec{a}'\rangle_E,
\]
where $\langle ,\rangle_E$ is defined as the ordinary dot inner product(or Euclidean inner product). It is not difficult to verify that every vector in $\F_q^{2n}$ is orthogonal to itself with symplectic inner product, i.e. for any $\vec{a}\in\F_q^{2n},\langle\vec{a},\vec{a}\rangle_S=0$. For an $\F_q$-linear code $\cC$ in $\F_q^{2n}$, define the symplectic dual of $\cC$ by
\[
\cC^{\perp_S}=\{(\vec{x}\mid\vec{y}):\langle(\vec{x}\mid\vec{y}),(\vec{a}\mid \vec{b})\rangle_S=0\ for\ all\ (\vec{a}\mid \vec{b})\in\cC \}.
\]
It is easy to verify that $dim_{\F_q}(\cC)+dim_{\F_q}(\cC^{\perp_S})=2n$. A code $\cC$ is said symplectic self-orthogonal if $\cC\subseteq \cC^{\perp_S}$, and self-dual if $\cC=\cC^{\perp_S}$. The rate $R$ of symplectic self-orthogonal code $\cC$ is $k/n$.

For a vector $\vec{v}=(a_1,\ldots,a_n\mid b_1,\ldots,b_n)\in \F_q^{2n}$ and an index $I\subseteq [n]$, denote by 
\[
\vec{v}_I := (a_i,b_i)_{i\in I}
\]
the truncation of $\vec{v}$ to the coordinates indexed by $I$. Then define the restricted support set and restricted symplectic weight:
\[
\supp_S(\vec{v}; I):= \{\, i\in I : (x_i, y_i)\neq (0,0) \,\}, \quad \wt_S(\vec{v};I):=|\supp_S(\vec{v};I)|.
\]
In particular, when $I=[n]$, this reduces to the classical definitions of the support set and the weight, denoted by $\supp_S(\vec{v})$ and $\wt_S(\vec{v})$.

For two vectors $\vec{u},\vec{v}\in \F_q^{2n}$, the symplectic distance is defined as
\[
d_S(\vec{u},\vec{v})=\wt_S(\vec{u}-\vec{v}).
\]
The symplectic minimum distance of a linear code $\cC\in \F_q^{2n}$ is defined as
\[
d_S(\cC):=min\{\wt_S(\vec{v}):\ \vec{v}\in \cC \setminus \{(\vec{0},\vec{0})\} \}.
\]
Similar to the Hamming distance case, let $B_q^{S}(2n,d):=\{\vec{x}\in \F_q^{2n}:\wt_S(\vec{x})\leq d \}$ denote the symplectic ball of radius $d$ and $\Vol_q^S(2n,d)$ denote the volume of this ball. It is clear that
\[
\Vol_q^S(2n,d)=\sum_{i=0}^{d}(q^2-1)^i\binom{n}{i}.
\]

A classical result on quantum code construction via symplectic self-orthogonal codes is as follows:
\begin{lemma}[Theorem~4, \cite{ashikhmin2002nonbinary}]\label{lem:construct-quantum-codes}
    If $\cC$ is a $q$-ary symplectic self-orthogonal $[2n,k]_q$ code, then there exists a $q$-ary $[[n,n-k,d]]$ quantum code with $d=d_S(\cC^{\perp_S})$.
\end{lemma}
There are two bounds which have been established as necessary conditions for quantum codes:
\begin{lemma}[Quantum Hamming Bound, \cite{calderbank1998quantum}]
    For any pure stabilizer quantum code $[[n,k,d]]_q$
    \[
    q^{n-k}\geq \sum_{i=0}^{\lfloor\frac{d-1}{2}\rfloor}(q^2-1)^{i}\binom{n}{i}.
    \]
\end{lemma}
\begin{lemma}[Quantum Singleton Bound, \cite{rains2002nonbinary}]
    For any quantum code $[[n,k,d]]_q$, $n\geq k+2d-2$.
\end{lemma}
We present the basic form of the quantum \emph{Gilbert-Varshamov} bound in Appendix~\ref{thm:basic-form-GVbound}. In particular, Feng an Ma \cite{feng2004finite} improved \emph{Gilbert-Varshamov} bound based on the Hermitian construction via $\F_{q^2}$-linear Hermitian self-orthogonal codes, which is a sufficient condition for the existence of quantum codes.


\begin{lemma}[Theorem~1.4, \cite{feng2004finite}]\label{lem:feng-GV-bound}
Suppose that $n>k\ge 2$, $d\ge 2$, and $k$ is even. Then there exists a pure stabilizer quantum code $[[n,\,n-k,\,d]]_q$ provided that
\[
\sum_{i=0}^{d-1}\binom{n}{i}(q^2-1)^i<q^{k+2}.
\]
\end{lemma}

\begin{rmk}
By contrast, Lemma~\ref{lem:feng-GV-bound} yields the sufficient condition $\sum_{i=0}^{d-1}\binom{n}{i}(q^2-1)^i<q^{k+2}$. For fixed $q$ and $\delta=d/n$ and a constant $c$ depend on $\delta$, the ratio between the condition and our main result equals
\[
\frac{c\cdot (q-1)}{q^2}\cdot \frac{\sqrt n}{1-q^{-(2n-k)}} ,
\]
which grows on the order of $\sqrt n$ as $n\to\infty$. Hence, for sufficiently large $n$, Corollary~\ref{cor:quantum-codes} provides a strictly weaker criterion than Lemma~\ref{lem:feng-GV-bound}.
\end{rmk}

Next, we briefly introduce some fundamental concepts and results from probability theory. Let $[n]$ denote the set $\{1,\ldots,n\}$. For a set $A$ we denote by $\binom{A}{\ell}$ the family of all subsets of $A$ with $\ell$ elements, and similarly $\binom{X}{\leq \ell}$ denotes the collection of all subsets of size up to $\ell$ in $X$. Let $A_1,\ldots,A_n$ be $n$ events. For any subset $X\subseteq [n]$, we denote by $A_X:=\cap_{i\in X}A_i$. The inclusion-exclusion principle is used to estimate the size of a union of multiple sets. In particular, given $n$ events $A_1,\ldots,A_n$, we have
\[
\Pr[\cup_{i=1}^{n}A_i]=\sum_{X\in \binom{[n]}{1}}\Pr[A_X]-\sum_{X\in \binom{[n]}{2}}\Pr[A_X]+\ldots+(-1)^{n+1}\Pr[\cap_{i=1}^{n}A_i].
\]
This alternating sum provides either an exact expression (if all terms are included) or upper/lower bounds when truncated, as in the Bonferroni inequalities.
\begin{lemma}[Bonferroni inequalities]
    Let $A_1,\ldots,A_n$ be events in a probability space. For any integer $1\leq i\leq n$, define
    \[
    S_i=\sum_{X\in \binom{[n]}{i}}\Pr[A_X].
    \]
    Then the probability $\Pr[\cup_{i=1}^{n}A_i]$ satisfies the following inequalities:
    \[
    \sum_{i=1}^{2k}(-1)^{i-1}S_i\leq \Pr[\cup_{i=1}^{n}A_i]\leq \sum_{i=1}^{2k+1}(-1)^{i-1}S_i
    \]
    for any integer $1\leq k\leq \frac{n-1}{2}$.
\end{lemma}

We also need the Chernoff bound to bound the sum of independent random variables.
\begin{lemma}[Chernoff Bound]
    Let $X_1,\ldots,X_n$ be independent random variables taking value in $[0,1]$, and let $X=\sum_{i=1}^{n}X_i$ with expectation $\mu=\mathbb{E}[X]$. Then for any $0<\delta <1$, the following multiplicative Chernoff bounds hold:
    \begin{itemize}
        \item Upper tail:
        \[
        \Pr[X\geq (1+\delta)\mu]\leq exp\left( -\frac{\delta^2\mu}{2+\delta} \right).
        \]
        \item Lower tail:
        \[
        \Pr[X\leq (1-\delta)\mu]\leq exp\left(-\frac{\delta^2\mu}{2} \right).
        \]
    \end{itemize}
    These inequalities imply that the sum $X$ is concentrated around its expectation $\mu$, and large deviations are exponentially unlikely.
    \label{Chernoff}
\end{lemma}

\section{An Initial Improvement of the Gilbert-Varshamov Bound Using Bonferroni Inequalities}
In this section, we will improve GV bound with a constant multiplcative factor by using Bonferroni inequalities. In next section, we will refine our argument which leads to our $\Omega(\sqrt(n))$ multiplicative factor improvement to the GV bound.

First, we provide a brief characterization of a property of random linear codes.

\begin{lemma}[Lemma~3.1.12, \cite{guruswami2012essential}]
    Given a non-zero vector $\vec{m}\in \F_q^k$ and a uniformly random $k\times n$ matrix $G$ over $\F_q$, the vector $\vec{m}\cdot G$ is uniformly distributed over $\F_q^n$.
\end{lemma}

Let $\cC$ be an random linear code of dimension $k$ and length $n$ with generator matrix $G$. We index the nonzero codeword in $\cC$ by $\vec{c}_m=\vec{m}G$ where $\vec{m}\in \F_q^k$ is a nonzero vector. In this sense, the codewords $\vec{c}_{m_1},\ldots, \vec{c}_{m_\ell}$ are linearly dependent if and only if $\vec{m}_1,\ldots,\vec{m}_\ell$ are linearly dependent. Let $B_q(n,d)$ denote the Hamming ball of radius $d$. We next show that the probability that the linear combination of $\ell$ vectors, chosen uniformly at random from a Hamming ball, remains entirely within the same Hamming ball is exponentially small.

\begin{lemma}[Proposition~3.3.3, \cite{guruswami2012essential}]
    Let $q\geq 2$ be an integer and $0\leq p\leq 1-\frac{1}{q}$ be a real number. Then:
    \begin{enumerate}[label=(\roman*)]
        \item $\Vol_q(n,pn)\leq q^{H_q(p)n}$; and
        \item for large enough $n$, $\Vol_q(n,pn)\geq q^{H_q(p)n-o(n)}$.
    \end{enumerate}
    where $H_q(x)=x\log_q(q-1)-x\log_qx-(1-x)\log_q(1-x)$.
    \label{lem:Hamming volume}
\end{lemma}

\begin{lemma}\label{lem:Hamming ball}
    Assume $d-1= \delta n$ with some constant $\delta\in (0,1-\frac{1}{q})$. Let $\vec{v}_1,\ldots \vec{v}_\ell$ be $n\geq \ell \geq 2$ random vectors distributed uniformly at random in the Hamming ball $B_q(n,\delta n)\subseteq \F_q^n$. Then, the probability that $\sum_{i=1}^{\ell}\vec{v}_i\in B_q(n,\delta n)$ is at most $2^{-h_{\delta}n}$ for some constant $h_{\delta}\in (0,1)$.
\end{lemma}

\begin{proof}
    By \cref{lem:Hamming volume}, the number of codeword with weight $r$ is $\binom{n}{r}(q-1)^r=2^{\Theta(H_q(\frac{r}{n}))}$ where the entropy function $H_q(x)$ is monotone increasing in the domain $x\in (0,1-\frac{1}{q})$. Let $\varepsilon>0$, this implies that if $\vec{v}$ is a random vector distributed uniformly at random in $B_q(n,d)$, the probability that $\wt(\vec{v})$ is less than $(\delta-\varepsilon)n$ is at most $2^{-a_{\delta}n}$ for some constant $a_\delta\in (0,1)$. We now prove this lemma by induction. Assume $\ell=2$ and we bound the probability that $\wt(\vec{v}_1+\vec{v}_2)< d$. With probability at most $2^{-a_{\delta}n+1}$, the weight of $\vec{v}_1$ and $\vec{v}_2$ are at most $(\delta-\varepsilon)n$. It follows that
    \begin{equation}
    \label{eq:two vector}
    \begin{aligned}
    &\Pr[\wt(\vec{v}_1 + \vec{v}_2) < d] 
    \leq 2 \Pr[\wt(\vec{v}_1) < (\delta - \varepsilon)n] + \Pr[\wt(\vec{v}_i) \geq (\delta - \varepsilon)n,\, i \in [2]]  \\
    & \cdot  \Pr[\wt(\vec{v}_1 + \vec{v}_2) < d 
    \mid \wt(\vec{v}_i) \geq (\delta - \varepsilon)n,\, i \in [2]].
    \end{aligned}
    \end{equation}
    Now, we suppose that $\wt(\vec{v}_1)=\delta_1n,\ \wt(\vec{v}_2)=\delta_2n,\ (\delta-\varepsilon)\leq \delta_1,\delta_2\leq \delta$, consider the same probability when each coordinate of $\vec{v}_i$ is chosen to be non-zero element with probability $\delta_i/n$ independently. Thus, the expectation of the Hamming weight at the $j$-th coordinate of $\vec{v}_i$ is given by
    \[
    \mathbb{E}[\wt(\vec{v}_i)_j]=\delta_i,\ for\ \ 1\leq j\leq n.
    \]
    Therefore, the expectation of $\wt(\vec{v}_1+\vec{v}_2)_j$ can be computed as:
    \[
    \begin{aligned}
        \mathbb{E}(\wt(\vec{v}_1+\vec{v}_2)_j)&=\delta_1(1-\delta_2)+\delta_2(1-\delta_1)+\frac{q-2}{q-1}\delta_1 \delta_2 \\
        &=\delta_1+\delta_2-\frac{q}{q-1}\delta_1 \delta_2 \\
        &\geq 2(\delta-\varepsilon)-\frac{q}{q-1}(\delta-\varepsilon)^2.
    \end{aligned}
    \]
    As for any $\delta\in (0,1-\frac{1}{q})$, $2\delta-\frac{q}{q-1}\delta^2>\delta$, there exists some constant $\varepsilon>0$ such that
    \[
    \mathbb{E}(\wt(\vec{v}_1+\vec{v}_2)_j)\geq 2(\delta-\varepsilon)-\frac{q}{q-1}(\delta-\varepsilon)^2\geq (\delta+\varepsilon).
    \]
    Since $\mathbb{E}[\wt(\vec{v}_1+\vec{v}_2)]=\sum_{j=1}^{n}\mathbb{E}[\wt(\vec{v}_1+\vec{v}_2)_j]\geq (\delta+\varepsilon)n$, by \cref{Chernoff}, the probability that $\wt(\vec{v}_1+\vec{v}_2)$ is less than $\delta n$ is at most $2^{-b_{\delta}n}$ for some constant $b_\delta\in (0,1)$. We take $c_\delta=min\{a_\delta, b_\delta \}$ and it follows that $\cref{eq:two vector}$ can be proved to be smaller than $2^{-c_{\delta}n}$.
    
    Now we use the induction to prove that the probability of $wt(\sum_{i=1}^{\ell}\vec{v}_i)\leq \delta n$ is at most $2^{-h_\delta'n}$. Assume this holds for $\ell-1$, by induction, we can suppose that the probability of $wt(\sum_{i=1}^{\ell-1}\vec{v}_i)\geq (\delta-\varepsilon)n$ is at least $1-2^{-h_{\delta}n}$. From above argument, we also know that with probability at most $2^{-c_{\delta}n}$, $wt(\vec{v}_\ell)\leq (\delta-\varepsilon)n$. Then, we fix the vector $\vec{v}:=\sum_{i=1}^{\ell-1}\vec{v}_{i}$ and apply the $\ell=2$ argument to $\vec{v}+\vec{v}_\ell$ conditioning that both of the vector has weight at least $(\delta-\varepsilon)n$. Note that the expectation $\mathbb{E}[\wt(\vec{v})]\geq (\delta-\varepsilon)n$, following the derivation in the case of $\ell=2$, we obtain the corresponding result that the probability of $\wt(\vec{v}+\vec{v}_\ell)< d)$ is at most $2^{-c_{\delta}n}$. By union bound, we complete the induction and there exists $h_\delta'\in (0,1)$ such that $\Pr[\wt(\sum_{i=1}^{\ell}\vec{v}_i)< d]\leq 2^{-h_{\delta}'n}$. The proof is completed.
\end{proof}
We are ready to prove the improved GV bound. As a warm up, we consider a simple case. Let $E_{\vec{m}}$ be event that a random vector $\vec{c}_{\vec{m}}$ has Hamming weight at most $d-1$ if no confusion occurs. Let $W\subseteq (\F_q^k-\{\vec{0}\})/\F_q^*$, i.e., $W=\bigcup_{i=1}^{k-1}W_i$ where $W_i=\{(0,\ldots,0,1,\vec{x}):\vec{x}\in \F_q^i \}$.

\begin{thm}
    Let $d=\delta n$ for some $\delta \in(0,1-\frac{1}{q})$. There exists $[n,k,d]$-linear code over $\F_q$ if 
    \[
    \frac{q^k-1}{q-1}<\frac{1.7q^n}{\sum_{i=0}^{d-1}\binom{n}{i}(q-1)^i}.
    \]
\end{thm}

\begin{proof}
    Let $\cC$ be an random linear code of length $n$ and dimension $k$. $\vec{c}_m$ be the codeword in $\cC$ for $\vec{m}\in \F_q^k$. It suffices to prove that all random vector $\vec{c}_m$ for $\vec{m}\in W$ are of weight at least $d$. Recall that $E_{\vec{m}}$ is the event that a random vector $\vec{c}_m$ has weight less than $d$ and $E_X$ is the event that all random vectors $\vec{c}_m$ for $\vec{m}\in X$ have weight less than $d$. We want to bound the probability $\Pr[\cup_{\vec{m}\in W}E_{\vec{m}}]$. By Bonferroni inequalities, we have:
    \[
        \Pr[\bigcup_{\vec{m}\in W}E_{\vec{m}}]\leq \sum_{\vec{m}\in W}\Pr[E_{\vec{m}}]-\sum_{X\in \binom{W}{2}}\Pr[E_X]+\sum_{X\in \binom{W}{3}}\Pr[E_X].
    \]
    It is clear that $a:=\Pr[E_{\vec{m}}]=\frac{\Vol_q(n,d-1)}{q^n}=\sum_{i=0}^{d-1}\binom{n}{i}(q-1)^i$. Next we proceed to bound the probability $E_X$ for $X=\{\vec{c}_x,\vec{c}_y \}$. Since $\vec{x},\vec{y}\in W$, they are $\F_q$-linearly independent. This implies the random vector $\vec{c}_x,\vec{c}_y$ are mutually independent and thus $\Pr[E_X]=a^2$. Then, we have
    \[
    \sum_{X\in \binom{W}{2}}\Pr[E_X]=\binom{|W|}{2}a^2.
    \]
    It remains to bound $\Pr[E_X]$ for $X=\{\vec{c}_x,\vec{c}_y,\vec{c}_z \}$. Clearly, if $\vec{x},\vec{y},\vec{z}$ are linearly independent, then $\Pr[E_X]=a^3$. Otherwise, since $\vec{x},\vec{y},\vec{z}\in W$, they must span a $2$-dimensional space. Then, we assume $\vec{z}=\lambda_1\vec{x}+\lambda_2\vec{y}$ for some nonzero $\lambda_1,\lambda_2$. By Lemma~\ref{lem:Hamming ball}, the probability $\Pr[E_X]\leq 2^{-h_{\delta}n}\cdot a^2$. Since there are at most $\binom{|W|}{2}(q-1)^2$ such $X$, we conclude 
    \[
    \sum_{X\in \binom{W}{3}}\Pr[E_X]\leq \binom{|W|}{3}a^3+\binom{|W|}{2}(q-1)^2\cdot a^2\cdot 2^{-h_\delta n}.
    \]
    Let $b=a|W|$ and the linear code exists if
    \[
    \Pr[\bigcup_{\vec{m}\in W}E_{\vec{m}}]\leq b-\frac{b^2}{2}(1-2^{-h_\delta n})+\frac{b^3}{6}<1
    \]
    holds for large enough $n$. Hence, we can choose $b=1.7$ to ensure the inequality holds. We conclude that if
    \[
    \frac{q^k-1}{q-1}<\frac{1.7q^n}{\sum_{i=0}^{d-1}\binom{n}{i}(q-1)^i},
    \]
    there exists $[n,k,d]$-linear code over $\F_q$. The proof is completed.
\end{proof}

\section{A General Improvement of the Gilbert–Varshamov Bound}
In the previous section, we presented a constant factor improvement of \emph{Gilbert-Varshamov} bound. In this section, we will expand the inclusion--exclusion sum of 
$\Pr\left[\bigcup_{\vec{m} \in W} E_{\vec{m}}\right]$ up to $\Omega(\sqrt{n})$-th term. This allows us to prove a $\Omega(\sqrt{n})$ multiplicative improvement for the \emph{Gilbert-Varshamov} bound. We start with a simple but very useful observation. Let $t = c\sqrt{n}$ be an odd integer for some small constant $c$. By Bonferroni inequalities, we have
\[
\Pr[\bigcup_{\vec{m} \in W} E_{\vec{m}}] 
\leq \sum_{i=1}^t (-1)^{i-1} \sum_{X \in \binom{W}{i}} \Pr[E_X].
\]
If all vectors in $X$ are linearly independent, $\Pr[E_X] = a^{|X|}$ for $a = \frac{\Vol_q(n,d-1)}{q^n}$. Assume this is true\footnote{Clearly, we will meet a large amount of set $X$ that such condition does not hold. Fortunately, this argument still holds if $X$ is relatively small.} for all 
$X \in \binom{W}{i}$ for $i \in [t]$ and we have
\[
\Pr[\bigcup_{\vec{m} \in W} E_{\vec{m}}] 
\leq -\sum_{i=1}^t \binom{|W|}{i} (-a)^i 
= - \sum_{i=1}^t \frac{(-b)^i}{i!}.
\]
We want this sum to be less than $\mathbf{1}$ and try to prove the following:
\[
\sum_{i=0}^t \frac{(-b)^i}{i!} 
\geq e^{-b} - \frac{b^{t+1}}{(t+1)!} > 0.
\]
The first inequality comes from the Taylor expansion of $e^{-x}$. Applying the Stirling formula to bound $b^{t+1}(t+1)!\leq (\frac{eb}{t+1})^{t+1}\leq e^{-(t+1)}$, we can argue that this holds for $b\leq \frac{t+1}{e^2}$ which completes the argument. This observation leads to our following theorem.
\begin{thm}
    Let $d=\delta n$ for some $\delta\in (0,1-\frac{1}{q})$. There exists $[n,k,d]$-linear code over $\F_q$ if
    \[
    \frac{q^k-1}{q-1}<\frac{c_\delta\sqrt{n}q^n}{\sum_{i=0}^{d-1}\binom{n}{i}(q-1)^i}
    \]
    for some constant $c_{\delta}>0$.
\end{thm}

\begin{proof}
Let $\cC$ be a random linear code of length $n$ and dimension $k$. $\vec{\bc}_m$ be the codeword in $\cC$ for $\vec{m} \in \F_q^k$. It suffices to prove that all random vector $\vec{\bc}_m$ for $\vec{m} \in W$ are of weight at least $d$. By Bonferroni inequalities, we have
\begin{equation}
    \Pr[\bigcup_{\vec{m} \in W} E_{\vec{m}} ]
    \leq \sum_{i=1}^t (-1)^{i-1} \sum_{X \in \binom{W}{i}} \Pr[E_X].
\end{equation}
for some odd integer $t = \sqrt{h_{\delta} \log_q 2\cdot \frac{n}{2}}$ where $h_{\delta}$ is given in Lemma~\ref{lem:Hamming ball}.  

Let $X \subseteq W$ be any subset of size $|X|=r \leq t$. If all vectors in $X$ are linearly independent, we have $\Pr[E_X] = a^r$ with 
$a = \frac{\Vol_q(n,d-1)}{q^n}$. Otherwise, assume that $X = \{\vec{x}_1, \ldots, \vec{x}_r\}$ where $\vec{x}_1, \ldots, \vec{x}_s\ (s<r)$ are the maximally linearly independent vectors in $X$. Then we have
\[
\Pr[E_X] = \Pr[E_{\{\vec{x}_1, \ldots, \vec{x}_s\}}] \Pr[ E_X \mid E_{\{\vec{x}_1, \ldots, \vec{x}_s\}} ] 
\leq a^s\cdot \Pr[ E_{\vec{x_r}} \mid E_{\{\vec{x}_1, \ldots, \vec{x}_s\}} ] \leq a^s\cdot 2^{-h_{\delta} n},
\]
The last inequality is due to Lemma~\ref{lem:Hamming ball} as $\vec{\bc}_{x_1}, \vec{\bc}_{x_2}, \ldots, \vec{\bc}_{x_s}$ are distributed uniformly at random in $B(\vec{0}, d-1)$. Thus, for any linearly independent vectors $Y := \{\vec{x}_1, \ldots, \vec{x}_s\}$, we have the followings:
\[
\sum_{i=s}^t \sum_{\substack{Y \subseteq X \subseteq span\{Y\}, \\ |X| = i}} \Pr[E_X] \leq a^s + a^s \sum_{i=s+1}^t \binom{q^s}{i} 2^{-h_{\delta} n} 
\leq a^s \left( 1 + \frac{q^{t^2+1}}{t!} 2^{-h_{\delta} n} \right) 
\leq a^s (1 + 2^{-h_{\delta}\cdot \frac{n}{2}}),
\]
as $t = \sqrt{h_{\delta} \log_q 2\cdot \frac{n}{2}}$.

Let $\mathcal{X}_i \subseteq \binom{W}{i}$ be the collection of the set of vectors 
$\{\vec{x}_1, \ldots, \vec{x}_i\}$ that are linearly independent, and define 
$\bar{\mathcal{X}}_i = \binom{W}{i} \setminus \mathcal{X}_i$.  
Then, we have
\[
\sum_{i=1}^t (-1)^{i-1} \sum_{X \in \binom{W}{i}} \Pr[E_X]
= \sum_{i=1}^t (-1)^{i-1} \sum_{X \in \mathcal{X}_i} \Pr[E_X] 
+ \sum_{i=1}^t (-1)^{i-1} \sum_{X \in \bar{\mathcal{X}}_i} \Pr[E_X]
\]
\[
\leq \sum_{i=1}^{t}(-1)^{i-1}\sum_{X\in \mathcal{X}_i}\Pr[E_X]+\sum_{i=1}^{t}\sum_{X\in \bar{\mathcal{X}_i}}\Pr[E_X]
\leq \sum_{i=1}^t (-1)^{i-1} a^i \sum_{X \in \mathcal{X}_i} 
\left( 1 + (-1)^{i-1} 2^{-h_{\delta} n/2} \right).
\]
The size of $\mathcal{X}_i$ is at least 
\[
\frac{1}{i!} \prod_{j=1}^i (|W| - q^{j-1}).
\]
Thus, we have
\[
1\geq \frac{|\mathcal{X}_i|}{\binom{|W|}{i}} \ge 1 - \frac{\sum_{j=1}^i q^{j-1}}{|W|} \ge 1 - q^{-k+t}.\ \ \ \ (\text{Since }i\leq t\text{ and }|W|=\frac{q^k-1}{q-1})
\]
This implies that
\begin{align*}
(-1)^{i-1} a^i \sum_{X \in \mathcal{X}_i} \left( 1 + (-1)^{i-1} 2^{-h_{\delta}\cdot \frac{n}{2}} \right) 
&=a^i|\mathcal{X}_i|\left( 1 + (-1)^{i-1} 2^{-h_{\delta}\cdot \frac{n}{2}} \right) \\
&\leq a^i \binom{|W|}{i} (1 + 2^{-h_{\delta}\cdot \frac{n}{2}}),
\end{align*}
for odd $i$ and
\begin{align*}
(-1)^{i-1} a^i \sum_{X \in \mathcal{X}_i} \left( 1 + (-1)^{i-1} 2^{-h_{\delta}\cdot \frac{n}{2}} \right) 
&\leq -a^i\left(\binom{|W|}{i}-q^{-k+t}\binom{|W|}{i} \right)\left(1-2^{-h_{\delta}\cdot \frac{n}{2}} \right) \\
&\leq - a^i \binom{|W|}{i} + a^i \binom{|W|}{i} (2^{-h_{\delta}\cdot\frac{n}{2}} + 2^{-\frac{k}{2}}) \\
&(\text{Since }k=Rn,\text{$R$ is the rate of code $\cC$})\\
&\leq - a^i \binom{|W|}{i} + a^i \binom{|W|}{i} 2^{-h_{\delta}'\cdot \frac{n}{2} + 1}
\end{align*}
for even $i$ where $h_{\delta}'=max\{\frac{h_{\delta}}{2},\frac{k}{2n} \}$. In summary, we obtain:
\[
\Pr[\bigcup_{\vec{m}\in W}E_{\vec{m}}]\leq \sum_{i=1}^{t}(-1)^{i-1}a^i\binom{|W|}{i}+\sum_{i=1}^{t}a^i\binom{|W|}{i}2^{-h_{\delta}'n+1}.
\]
Define $b=|W|a$ and the right hand side becomes
\[
\sum_{i=1}^{t}(-1)^{i-1}\frac{b^i}{i!}+\sum_{i=1}^{t}\frac{b^i}{i!}2^{-h_{\delta}'n+1}\leq 1+\frac{b^{t+1}}{(t+1)!}-e^{-b}+e^{b}2^{-h_{\delta}'n+1}.
\]
This inequality is due to the Taylor expansion of $e^x$ and $e^{-x}$:
\[
e^b\geq \sum_{i=0}^{t}\frac{b^i}{i!},\ e^{-b}\geq \sum_{i=0}^{t+1}(-1)^i\frac{b^i}{i!}.
\]
We concludes the theorem by choosing $b=\frac{t+1}{e^2}=\Omega(\sqrt{n})$. The proof is completed.
\end{proof}

\section{Extending to the Quantum Gilbert-Varshamov Bound}
In the previous section, we presented an improvement of the classical \emph{Gilbert-Varshamov} bound. In this section, we extend our approach to the quantum setting and establish an improved quantum \emph{Gilbert-Varshamov} bound.

Suppose $\cC$ is a random symplectic self-orthogonal code of length $2n$ and dimension $k$ with parity-check matrix $H$. Let $\cC^{\perp_S}$ denote its orthogonal complement, which is a linear code with parameters $[2n,2n-k]_q$ and $H$ serves as its generator matrix. By \cref{lem:construct-quantum-codes}, we note that the minimum distance of quantum code is completely determined by its dual code instead of code itself. Thus, we need to bound the minimum weight of nonzero codes in $\cC^{\perp_S}$ instead of $\cC$. Similar to the Hamming case, we index the nonzero codeword in $\cC^{\perp_S}$ by $\vec{c}_m=\vec{m}H$ where $\vec{m}\in \F_q^{2n-k}$ is a nonzero vector. In this setting, the codewords $\vec{c}_{m_1},\ldots, \vec{c}_{m_{\ell}}$ are linearly dependent if and only if $\vec{m}_1,\ldots,\vec{m}_{\ell}$ are linearly dependent. We use $E_{\vec{m}}$ to denote the event that a random vector $\vec{c}_{\vec{m}}$ has symplectic weight at most $d-1$ if no confusion occurs. We can define a set $W\subseteq \F_q^{2n-k}\setminus\{\vec{0}\}$ such that $W$ is of size $\frac{q^{2n-k}-1}{q-1}$ and for any $\vec{x}\in \F_q^{2n-k}\setminus \{\vec{0}\}$, there exists a $\vec{y}\in W$ with $\vec{x}=\lambda \vec{y}$ for some nonzero $\lambda\in \F_q$. In fact, we can explicitly write $W=\bigcup_{i=1}^{2n-k-1}W_i$ where $W_i=\{(0,\ldots,0,1,\vec{x}):\vec{x}\in \F_q^i \}$.  

In contrast to the Hamming distance case, the selection of $\ell$ linearly independent vectors in the symplectic code cannot be treated as a uniform random sampling from the entire space. Moreover, the analysis of the probability that their linear combinations lie outside the symplectic ball requires a different approach. To address this challenge, We begin by computing $\Bigl|B_q^S(2n,d-1)\cap\langle\vec{u}_1\rangle^{\perp_S}\cap\ldots\cap\langle\vec{u}_\ell\rangle^{\perp_S}\Bigr|$ , which becomes the key ingredient of our method.

\begin{lemma}\label{lem:one-vec-bound}
    Let $\vec{u}$ be an nonzero random vector in $B_q^S(2n,d-1)$, then we have the bound:
    \[
    \frac{1-q^{-\Omega(n)}}{q}\Vol_q^S(2n,d-1)\leq \Bigl|B_q^S(2n,d-1)\cap\langle\vec{u}\rangle^{\perp_S}\Bigr|\leq \frac{1+q^{-\Omega(n)}}{q}\Vol_q^S(2n,d-1).
    \]
\end{lemma}
\begin{proof}
    Let $\vec{u}=(a_1,\ldots,a_n\mid b_1,\ldots,b_n)\in \F_q^{2n}$ and $\wt_S(\vec{u})=t<d$. Without loss of generality, we assume that $\wt(\vec{u};[t])=t$. If $\vec{v}=(x_1,\ldots,x_n\mid y_1,\ldots,y_n)\in \langle\vec{u}\rangle^{\perp_S}\cap B_q^S(2n,d-1)$, it follows that $\wt_S(\vec{v})\leq d-1$ and:
    \begin{equation}\label{eq:one-vec-orth}
        \sum_{i=1}^{t}(a_iy_i-b_ix_i)=0.
    \end{equation}
   Let $S$ be the set of all vectors $\vec{v} = (x_1,\ldots,x_t \mid  y_1,\ldots,y_t) \in \F_q^{2t}$ satisfying \cref{eq:one-vec-orth} and $N_h=\{\vec{v}\in S:\operatorname{wt}_S(\vec{v},[t])=h\}$. Let $A_j=\{\vec{v}\in S:(x_j,y_j)= (0,0)\}\subseteq S$, i.e. the collection of vectors in $S$ whose $j$-th coordinate are $(0,0)$.
   \noindent
    By inclusion-exclusion principle, we have:
     \begin{align*}
        |N_h|&=|S|-\sum_{i=1}^h|A_i|+\sum_{1\leq i_1<i_2\leq h}|A_i\cap A_j|+\ldots+|A_1\cap A_2\cap \ldots\cap A_h|\\
        &=\sum_{j=0}^{h-1}(-1)^j\binom{h}{j}q^{2h-1-2j}+(-1)^h.
    \end{align*}
    \noindent
    Therefore,
    \begin{equation*}
    N_h\leq \sum_{j=0}^{h-1}(-1)^j\binom{h}{j}q^{2h-1-2j}+1=q^{2h-1}(1-\frac{1}{q^2})^h+1.
    \end{equation*}
    If $h=\Omega(n)$, we have
    \[
    \begin{aligned}
        (q^2-1)^{w-h}\cdot N_h&\leq (q^2-1)^{w-h}\cdot (q^{2h-1}(1-\frac{1}{q^2})^h+1) \\
        &=\frac{1}{q}(q^2-1)^w+(q^2-1)^{w-h} 
        \leq \frac{1+q^{-\Omega(n)}}{q}(q^2-1)^w.
    \end{aligned}
    \]
    To estimate the size of $B_q^S(2n,d-1)\cap \langle\vec{u}\rangle^{\perp_S}$,
    we first select $h$ positions within the first $t$ entries, and $\wt_S(\vec{v})-h$ positions from the remaining entries. Therefore,
    \begin{equation}\label{eq:one-vec-cap}
        |B_q^S(2n,d-1)\cap \langle\vec{u}\rangle^{\perp_S}|=\sum_{w=0}^{d-1}\sum_{h=0}^{min\{w,t\}}\binom{t}{h}\binom{n-t}{w-h}(q^2-1)^{w-h}\cdot N_h.
    \end{equation}
    We now proceed to estimate $|B_q^S(2n,d-1)\cap\langle\vec{u}\rangle^{\perp_S}|$. For $h=o(n)$, the contribution is at most $\Vol_q^S(2n-2t,d-h)\cdot \binom{t}{h}(q^2-1)^{h}$. As $t,d=\Theta (n)$, this term is exponentially smaller than the dominant term in \cref{eq:one-vec-cap}. Thus, it suffices to bound the term in \cref{eq:one-vec-cap} for the term with $h=\Omega(n)$.
    Hence,
    \begin{align*}
        &\sum_{w=0}^{d-1}\sum_{h=0}^{min\{w,t\}}\binom{t}{h}\binom{n-t}{w-h}(q^2-1)^{w-h}\cdot N_h\\
        &\leq \frac{1+q^{-\Omega(n)}}{q}\sum_{w=0}^{d-1}(q^2-1)^w\sum_{h=0}^{min\{w,t\}}\binom{t}{h}\binom{n-t}{w-h}\\
        &=\frac{1+q^{-\Omega(n)}}{q}\sum_{w=0}^{d-1}(q^2-1)^w\binom{n}{w} 
        =\frac{1+q^{-\Omega(n)}}{q}\Vol_q^S(2n,d-1).
    \end{align*}
    In the same way, we can obtain:
    \[
    N_h=\sum_{j=0}^{h-1}(-1)^j\binom{h}{j}q^{2h-1-2j}+(-1)^h\geq q^{2h-1}(1-\frac{1}{q^2})^h-1
    \]
    and
    \[
    |B_q^S(2n,d-1)\cap \langle\vec{u}\rangle^{\perp_S}|\geq \frac{1-q^{-\Omega(n)}}{q}\Vol_q^S(2n,d-1).
    \]
    Therefore, the proof is completed.
\end{proof}

We first assume that $\vec{u}_1,\ldots,\vec{u}_\ell$ be nonzero, linearly independent random vectors in $B_q^S(2n,d-1)$, and subsequently incorporate the mutual orthogonality relations among them in Corollary~\ref{col:l-vector-bound}. Our argument shows that if $\vec{x}\in B_q^S(2n,d-1)\cap \bigcap_{i=1}^{\ell}\langle\vec{u}_i\rangle^{\perp_S}$, the vector coordinates must satisfy $\ell$ linear equations. Let the coefficient matrix induced by the $\ell$ linear equations be $A$ and $\rank(A)=\ell$. We prove that the size of the intersection $B_q^S(2n,d-1)\cap \bigcap_{i=1}^{\ell}\langle\vec{u}_i\rangle^{\perp_S}$ is affected only when $\rank(A)<\ell$, and the total contribution due to such rank-deficiency events is exponentially small.

\begin{thm}\label{thm:l-vector-bound}
    Let $\ell=\sqrt{n}$, $\vec{u}_1,\ldots,\vec{u}_\ell$ be nonzero, linearly independent random vectors in $B_q^S(2n,d-1)$, we obtain the bound:
    \[
    \frac{1-q^{-\Omega(n)}}{q^\ell}\Vol_q^S(2n,d-1)\leq |B_q^S(2n,d-1)\cap \langle\vec{u}_1\rangle^{\perp_S}\cap \langle \vec{u}_2\rangle^{\perp_S}\ldots\langle\vec{u}_\ell\rangle^{\perp_S}|\leq \frac{1+q^{-\Omega(n)}}{q^\ell}\Vol_q^S(2n,d-1).
    \]
\end{thm}

\begin{proof}
    Write $\vec{u}_i=(a_{i1},\ldots,a_{in}\mid b_{i1},\ldots,b_{in})$ for $i\in [\ell]$. Any $\vec{x}=(x_1,\ldots, x_n \mid y_1,\ldots y_n)\in B_q^S(2n,d-1)\cap \bigcap_{i=1}^{\ell}\langle \vec{u}_i\rangle^{\perp_S}$ must satisfy the $\ell$ symplectic-orthogonality constraints:
\begin{equation}\label{eq:l-equations}
    \sum_{j=1}^n(a_{ij}y_j-b_{ij}x_j)=0,\qquad i\in[\ell].
\end{equation}

Let $A\in \F_q^{\ell\times n}$ be the coefficient matrix induced by $\cref{eq:l-equations}$ (each column corresponds to one symplectic pair $(x_j,y_j)$). Since $\vec{u}_1,\ldots,\vec{u}_{\ell}$ are linearly independent, $A$ has full row rank, i.e., $\rank(A)=\ell$. Moreover, $\vec{u}_1,\ldots,\vec{u}_{\ell}$ are random vectors in  $B_q^S(2n,d-1)$. We note that given that $\vec{u}_1,\ldots,\vec{u}_{\ell}$ distribute uniformly at random in  $B_q^S(2n,d-1)$, the probability $\vec{u}_1,\ldots,\vec{u}_{\ell}$ are linearly independent is at least $1-\exp(\Omega(n))$ if $d=\Omega(n)$. Thus, we will only assume that $\vec{u}_1,\ldots,\vec{u}_{\ell}$ are random vectors in  $B_q^S(2n,d-1)$ in the following argument. 

Partition $B_q^S(2n,d-1)$ into $d$ sets according to the symplectic weight:
\[
B_q^S(2n,d-1)=\bigcup_{w=0}^{d-1}\mathcal B_w,\qquad \mathcal B_w:=\{\vec{x}:\wt_S(\vec{x})=w\}.
\]

For vectors $\vec{x}$ of symplectic weight $w$, let $T\subseteq [n]$ denote the support set of $\vec{x}$. Restricting \cref{eq:l-equations} to the variables indexed by $T$ yields a reduced system with coefficient submatrix $A_T$ which is a submatrix of $A$ by restricting $A$ to the columns indexed by $T$. Denote by $N(T)$ the number of assignments of $(x_j,y_j)_{j\in T}\in (\F_q^2\setminus \{(0,0)\})^w$ satisfying \cref{eq:l-equations}. Then the intersection size can be written as:
\begin{equation}\label{eq:l-sum}
    \left|B_q^S(2n,d-1)\cap \bigcap_{i=1}^{\ell}\langle\vec{u}_i\rangle^{\perp_S}\right|=\sum_{w=0}^{d-1}\ \sum_{\substack{T\subseteq[n]\\|T|=w}} N(T).
\end{equation}

Fix $T\subseteq [n]$ with $|T|=\alpha n$. It is well known that the total number of vectors in $B_q^S(2n,d-1)$ with symplectic weight $\alpha <\frac{2\delta}{3}$ is at most a $q^{-\Omega(n)}$ fraction of $\Vol_q^S(2n,d-1)$. Hence, it suffices to consider $\alpha\in [\frac{2\delta}{3},\delta]$. 
we proceed our argument conditioning on this case. When $\rank(A_T)=\ell$, the $\ell$ constraints in \cref{eq:l-equations} behave as $\ell$ independent linear constraints on the $2\alpha n$ variables $(x_j,y_j)_{i\in T}$. Using the same inclusion-exclusion argument as in Lemma~\ref{lem:one-vec-bound}, one obtains
\[
N(T)=\frac{1\pm q^{-\Omega(n)}}{q^\ell}\,(q^2-1)^{\alpha n}
\qquad\text{whenever }\rank(A_T)=\ell.
\]

It remains to bound the total contribution of supports $T$ for which $\rank(A_T)<\ell$. For any $T$, define the bad event $\mathcal E_T:=\{\rank(A_T)<\ell\}$. As $\vec{u}_1,\ldots,\vec{u}_\ell$ are sampled uniformly at random from $B_q^S(2n,d-1)$, it is standard that for any fixed $\varepsilon>0$, it holds that
\[
\Pr\left[(\delta-\varepsilon)n\leq \wt_S(\vec{u}_i)\leq \delta n\right]\geq 1-q^{-\Omega (n)}.
\]
for every $i\in [\ell]$.
Consequently, by a union bound, the above event holds for all $i\in [\ell]$ simultaneously with probability at least $1-\ell\cdot q^{-\Omega (n)}$. 

Assume that $\wt_S(\vec{u}_i)=w_i$. Then the support $\supp(\vec{u}_i)\subseteq [n]$ is uniformly distributed among all subsets of $[n]$ of size $w_i$. Hence,
\[
X_i:=\wt_S(\vec{u}_i') = |\supp(\vec{u}_i)\cap T|
\]
follows a hypergeometric distribution with mean $\mathbb{E}[X_i]=\alpha w_i$. Therefore, by Chernoff-Hoeffding bound for hypergeometric random variables, for any $\varepsilon>0$, there exists $c'>0$ such that for every $i\in [\ell]$,
\[
\Pr\big[\,|X_i-\alpha w_i|\ge \varepsilon n \big]\le q^{-c'n}.
\]
In particular, if $w_i\in[(\delta-\varepsilon)n,\delta n]$, we have
\[
\Pr\left[ X_i\in [(\alpha\delta-\varepsilon')n, (\alpha\delta+\varepsilon')n]\right]\geq 1-q^{-\Omega(n)}.
\]
 for any $\varepsilon'>0$.

Let $W_{i-1}=\mathrm{span}\{\vec{u}_1',\ldots,\vec{u}_{i-1}'\}$, so $|W_{i-1}|\le q^{i-1}\le q^{\ell-1}$. Note that the support $\supp(\vec{u}_i)$ is uniform among all subsets of $[n]$ with size $w_i$. Hence, conditioning on $X_i=h$, the punctured vector $\vec{u}_i'=\vec{u}_i|_T$ is uniform over the symplectic sphere $\mathcal S_h(T):=\{\vec{v}\in(\F_q^2)^{|T|}: \wt_S(\vec{v})=h\}$, whose cardinality is $|\mathcal S_h(T)|=\binom{\alpha n}{h}(q^2-1)^h$.
Thus, for any $i\ge 2$ and any fixed subspace $W_{i-1}$,
\[
\Pr[\vec{u}_i'\in W_{i-1}\mid X_i=h]\le \frac{|W_{i-1}|}{|\mathcal S_h(T)|}
\le \frac{q^{\ell-1}}{\binom{\alpha n}{h}(q^2-1)^h}=q^{-\Omega(n)}.
\]

Taking a union bound over $i=2,\ldots,\ell$, we conclude that
\[
\begin{aligned}
    \Pr(\mathcal E_T)&=\Pr\left[\{\vec{u}_i'\}_{i=1}^\ell \text{ are linearly dependent}\right] \\
    &\leq \ell\cdot \Pr[\vec{u}_\ell'\in W_{\ell-1}] = q^{-a_\delta n}
\end{aligned}
\]
where the constant $a_\delta>0$ depends on $\delta$.

Now it suffices to control the total contribution of rank-deficient supports:
\[
\sum_{w=0}^{d-1}\ \sum_{\substack{T\subseteq[n]\\ |T|=w}} \mathbf 1\{\mathcal E_T\}\,N(T).
\]

Therefore, we take expectation over $A$ and bound the expected bad contribution at a fixed weight $w$:
\begin{align}\label{eq:rank-deficient}
\mathbb{E}_A\!\left[\sum_{\substack{T\subseteq[n]\\ |T|=w}} \mathbf 1\{\mathcal E_T\}\,N(T)\right]
&=\sum_{\substack{T\subseteq[n]\\ |T|=w}} \mathbb{E}_A\!\left[\mathbf 1\{\mathcal E_T\}\,N(T)\right] \notag \notag   \\
&\leq \sum_{\substack{T\subseteq[n]\\ |T|=w}} \Pr_{A}(\mathcal E_T)\cdot \max_{A} N(T) \notag \\
&\leq \sum_{\substack{T\subseteq[n]\\ |T|=w}} \Pr_A(\mathcal E_T)\cdot (q^2-1)^w \notag \\
&\leq q^{-a_{\delta}n}\binom{n}{w}(q^2-1)^w.
\end{align}
Here we use that for any fixed $T$, 
\[
\mathbb{E}_A\!\left[\mathbf{1}\{\mathcal{E}_T\}N(T)\right]
\le \Pr_A(\mathcal{E}_T)\cdot \max_A N(T).
\]
Moreover, since $N(T)$ counts assignments $(x_j,y_j)_{j\in T}\in(\F_q^2\setminus\{(0,0)\})^w$ satisfying \cref{eq:l-equations}, we have the bound $\max_A N(T)\le (q^2-1)^w$.
Summing over $w=0,1,\ldots,d-1$ and let variable $Z\coloneqq\sum_{w=0}^{d-1}\sum_{\substack{T\subseteq[n]\\ |T|=w}} \mathbf 1\{\mathcal E_T\}\,N(T)$, we obtain:
\[
\mathbb{E}_A[Z]
\le q^{-a_{\delta}n}\sum_{w=0}^{d-1}\binom{n}{w}(q^2-1)^w \notag
= q^{-a_{\delta}n}\,\Vol_q^S(2n,d-1).
\]
Since the rank-deficient contribution $Z$ is nonnegative and $\mathbb{E}(Z)\leq q^{-a_{\delta}n}\cdot \Vol_q^S(2n,d-1)$, Markov’s inequality implies that  $$\Pr[Z\leq q^{-a_{\delta} n/2}\cdot \Vol_q^S(2n,d-1)]\geq 1-q^{-a_{\delta}n/2}.$$

Using \cref{eq:l-sum,eq:rank-deficient} and following the form of the previous theorem, we obtain
\[
\frac{1- q^{-\Omega(n)}}{q^\ell}\sum_{w=0}^{d-1}\binom{n}{w}(q^2-1)^w\leq \left|B_q^S(2n,d-1)\cap \bigcap_{i=1}^{\ell}\langle\vec{u}_i\rangle^{\perp_S}\right|\leq \frac{1+ q^{-\Omega(n)}}{q^\ell}\sum_{w=0}^{d-1}\binom{n}{w}(q^2-1)^w.
\]
This proof is complete.

\end{proof}

The above theorem does not completely fit for our random quantum codes argument. The reason is that we do not require that $\vec{u}_1,\ldots, \vec{u}_\ell$ should be orthogonal to each other which is necessary for quantum code construction. We will add this constraint in the following corollary. We show that the orthogonal constraint will not affect the bound in \cref{thm:l-vector-bound}.

\begin{cor}\label{col:l-vector-bound}
    Let $\ell=\frac{\sqrt{a_{\delta}n}}{2}$ where $a_{\delta}$ is defined in the proof of  \cref{thm:l-vector-bound}, and let $\vec{u}_1,\ldots,\vec{u}_{\ell}$ be nonzero, linearly independent random vectors in $B_q^S(2n,d-1)$ that are mutually symplectic orthogonal to each other. Then, it holds
    \[
    \frac{1-q^{-\Omega(n)}}{q^\ell}\Vol_q^S(2n,d-1)\leq |B_q^S(2n,d-1)\cap \langle\vec{u}_1\rangle^{\perp_S}\cap\ldots\langle\vec{u}_\ell\rangle^{\perp_S}|\leq \frac{1+q^{-\Omega(n)}}{q^\ell}\Vol_q^S(2n,d-1).
    \]
\end{cor}

\begin{proof}
    We first apply \cref{thm:l-vector-bound} to conclude
    \[
    |B_q^S(2n,d-1)\cap \langle\vec{u}_1\rangle^{\perp_S}\cap \langle \vec{u}_2\rangle^{\perp_S}\ldots\langle\vec{u}_h\rangle^{\perp_S}|\geq\frac{1-q^{-\Omega(n)}}{q^{h-1}}\Vol_q^S(2n,d-1).
    \]
    for $h=1,\ldots,\ell$. This implies that there are at least $\frac{1-q^{-\Omega(n)}}{q^{h}}\Vol_q^S(2n,d-1)$ different $\vec{u}_{h+1}$ lying in the orthogonal space of the space spanned by $\vec{u}_{1},\ldots,\vec{u}_h$. It follows that there exists at least
    \begin{equation}\label{eq:mutual_orthogonal}
    \prod_{r=1}^{\ell}\frac{1-q^{-\Omega(n)}}{q^{r-1}}\Vol_q^S(2n,d-1)\geq \frac{(1-q^{-\Omega(n)}) \Vol_q^S(2n,d-1)^\ell}{q^{\frac{\ell(\ell-1)}{2}}}
    \end{equation}
    number of such tuples $(\vec{u}_1,\ldots,\vec{u}_\ell)\in B_q^S(2n,d-1)^\ell$ such that $\vec{u}_1,\ldots,\vec{u}_\ell$ are mutually symplectically orthogonal.

    On the other hand, \cref{thm:l-vector-bound} (together with the Markov step at the end of its proof) implies that, for a uniformly random $\ell$-tuple $(\vec{u}_1,\ldots,\vec{u}_{\ell})\in B_q^S(2n,d-1)^{\ell}$, the probability that a uniformly $\ell$-tuples $(\vec{u}_1,\ldots,\vec{u}_{\ell})$ does not satisfy the conclusion of \cref{thm:l-vector-bound} is at most $q^{-a_{\delta}n/2}$. Therefore, the number of such $\ell$-tuples is at most $q^{-a_{\delta}n/2}\Vol_q^S(2n,d-1)$.

    Combining this estimate \cref{eq:mutual_orthogonal}, we obtain that even the above "bad" $\ell$-tuples are mutually symplectically orthogonal, the fraction of "bad" $\ell$-tuples is at most
    \[
    \frac{q^{\ell(\ell-1)/2}}{(1-q^{-\Omega(n)})^{\ell-1}}\cdot q^{-a_{\delta}n/2}\leq q^{-a_{\delta}n/4},
    \]
    as $\ell=\frac{\sqrt{a_{\delta}n}}{2}$. Therefore, the mutual symplectic orthogonality constraint only affects an exponentially small subset of $\ell$-tuples. Consequently, using the same approach in \cref{thm:l-vector-bound}, the following conclusion still holds:
    \[
    \frac{1-q^{-\Omega(n)}}{q^\ell}\Vol_q^S(2n,d-1)\leq |B_q^S(2n,d-1)\cap \langle\vec{u}_1\rangle^{\perp_S}\cap \ldots\langle\vec{u}_\ell\rangle^{\perp_S}|\leq \frac{1+q^{-\Omega(n)}}{q^\ell}\Vol_q^S(2n,d-1).
    \]
\end{proof}

With the above theorem established, we proceed to generalize Lemma~\ref{lem:Hamming ball} to the quantum setting. Prior to this extension, we prove the corresponding result for the case where $\ell$ vectors are randomly selected from $B_q^S(2n,d-1)$.

\begin{lemma}
     Let $q\geq 2$ be an integer and $0< \delta < 1-\frac{1}{q^2}$ be a real number. Then:
    \begin{enumerate}[label=(\roman*)]
        \item $\Vol_q^S(2n,\delta n)\leq (q^2)^{H_{q^2}(\delta)n}$; and
        \item for large enough $n$, $\Vol_q^S(2n,\delta n)\geq (q^2)^{H_{q^2}(\delta)n-o(n)}$.
    \end{enumerate}
    where $H_{q^2}(x)=x\log_{q^2}(q^2-1)-x\log_{q^2}x-(1-x)\log_{q^2}(1-x)$.
    \label{lem:quantum volume}
\end{lemma}

Analogous to the analysis for $|B_q^S(2n,d-1)\cap \langle\vec{u}_1\rangle^{\perp_S}\cap \ldots\cap\langle\vec{u}_\ell\rangle^{\perp_S}|$, we first impose the assumption that $\vec{v}_1,\ldots,\vec{v}_{\ell}$ are random vectors distributed uniformly in $B_q^S(2n,\delta n)$, and subsequently incorporate their mutual orthogonality relations in \cref{thm:l-vector-orth-bound}. The proof of the next lemma follows the same approach as Lemma~\ref{lem:Hamming ball}. Hence, the detailed argument is provided in the Appendix~\ref{proof:l-linear-combination}.”

\begin{lemma}
    Assume $d-1=\delta n$ with some constant $\delta\in (0,1-\frac{1}{q^2})$. Let $\vec{v}_1,\ldots, \vec{v}_{\ell}$ be $n\geq \ell\geq 2$ random vectors distributed uniformly in the symplectic ball $B_q^S(2n,\delta n)\subseteq \F_q^{2n}$. Then, the probability that $\sum_{i=1}^{\ell}\vec{v}_i\in B_q^S(2n,\delta n)$ is at most $2^{-h_{\delta}n}$ for some constant $h_{\delta}\in (0,1)$.
    \label{lem:quantum-ball}
\end{lemma}

Lemma~\ref{lem:Hamming ball} implies that there are at most $2^{-h_\delta n}\cdot \Vol_q^S(2n,\delta n)^\ell$ sets of $\ell$-vectors $\vec{v}_1,\ldots,\vec{v}_\ell$ such that $\sum_{i=1}^{\ell}\vec{v}_i\in B_q^S(2n,\delta n)$. Next, we prove that even all these $\ell$ sets of vectors are mutually orthogonal, the probability $\Pr\left[\sum_{i=1}^{\ell}\vec{v}_i\in B_q^S(2n,\delta n)\right]$ is still at most $2^{-h'_\delta n}$ for some constant $h'_\delta\in (0,1)$.

\begin{thm}\label{thm:l-vector-orth-bound}
    Assume $d-1=\delta n$ with some constant $\delta\in (0,1-\frac{1}{q^2})$. Let $\ell=\min\{\frac{\sqrt{a_{\delta} n}}{2},\sqrt{h_{\delta}\log_q2\cdot \frac{n}{2}}\}$ where $a_{\delta
}$ is defined in the proof of \cref{thm:l-vector-bound} and $h_{\delta}$ is defined in Lemma~\ref{lem:quantum-ball}. Let $\vec{v}_1,\ldots,\vec{v}_{\ell}$ be $\ell$ random vectors distributed uniformly at random in the symplectic ball $B_q^S(2n,\delta n)\subseteq \F_q^{2n}$ and they are mutually symplectic orthogonal. Then, the probability that $\sum_{i=1}^{\ell}\vec{v}_i\in B_q^S(2n,\delta n)$ is at most $2^{-h'_{\delta}n}$ for some constant $h'_\delta\in (0,1)$.
\end{thm}

\begin{proof}
    We define "bad $\ell$-vectors" set as 
    $$B=\{(\vec{v}_1,\ldots,\vec{v}_\ell)\in B_q^S(2n,d-1):\sum_{i=1}^{\ell}\vec{v}_i\in B_q^S(2n,d-1)\}.$$ By Lemma~\ref{lem:quantum-ball}, the size of "bad $\ell$-vectors" set is bounded by:
    \[
    |B|\leq \Vol_q^S(2n,d-1)^\ell\cdot 2^{-h_\delta n}.
    \]
    where the constant $h_\delta\in (0,1)$. By Corollary~\ref{col:l-vector-bound}, we conclude that there are at least
    \[
    \prod_{r=1}^\ell\frac{1-q^{-\Omega(n)}}{q^{r-1}}\Vol_q^S(2n,d-1)\geq \frac{(1-q^{-\Omega(n)})^{\ell-1}\cdot \Vol_q^S(2n,d-1)^\ell}{q^{\frac{\ell(\ell-1)}{2}}}
    \]
    mutually orthogonal $\ell$-tuples $(\vec{v}_1,\ldots,\vec{v}_\ell)$.
    \noindent
    As $\ell\leq \sqrt{h_\delta \log_q2\cdot \frac{n}{2}}$, we have
    \[
    \frac{q^{\frac{\ell(\ell-1)}{2}}}{(1-q^{-\Omega(n)})^{\ell-1}}\cdot 2^{-h_\delta n}\leq \frac{2^{\frac{h_\delta n}{4}}}{1-(\ell-1)q^{-\Omega(n)}}\cdot 2^{-h_\delta n}=2^{-h'_\delta n}
    \]
    for constant number $h'_\delta\in (0,1)$. Therefore, the probability that $\sum_{i=1}^\ell\vec{v}_i\in B_q^S(2n,\delta n)$ is at most $2^{-h'_{\delta}n}$ for some constant $h'_{\delta}\in (0,1)$.
\end{proof}

We are ready to present the improved quantum \emph{Gilbert-Varshamov} bound.

\begin{thm}\label{thm:quantum-GV-Bound}
    Let $d=\delta n$ for some $\delta\in (0,1-\frac{1}{q^2})$. There exists $[2n,2n-k,d]$-linear symplectic self-orthogonal dual code $\cC^{\perp_S}$ over $\F_q$ if
    \[
    \frac{q^{2n-k}-1}{q-1}<\frac{c_\delta \sqrt{n}\cdot q^{2n}}{\sum_{i=0}^{d-1}\binom{n}{i}(q^2-1)^i}
    \]
    for some constant $c_{\delta}>0$.
\end{thm}

\begin{proof}
    Let $\cC^{\perp_S}$ be a linear symplectic self-orthogonal dual code of length $2n$ and dimension $2n-k$. Let $\vec{c}_m$ be the codeword in $\cC^{\perp_S}$ with the index $\vec{m}\in \F_q^{2n-k}$. Recall $W\subseteq \F_q^{2n-k}\setminus\{\vec{0}\}$ of size $\frac{q^{2n-k}-1}{q-1}$ such that for any $\vec{x}\in \F_q^{2n-k}\setminus \{\vec{0}\}$, there exists a $\vec{y}\in W$ with $\vec{x}=\lambda \vec{y}$ for some nonzero $\lambda\in \F_q$. It suffices to prove that all random vector $\vec{c}_m$ for $\vec{m}\in W$ are of symplectic weight at least $d$. Recall that $E_{\vec{m}}$ is the event that a random vector $\vec{c}_m$ has symplectic weight less than $d$ and $E_X$ is the event that all random vectors $\vec{c}_m$ with $\vec{m}\in X$ have symplectic weight less than $d$. We want to bound the probability $\cup_{\vec{m}\in W}E_{\vec{m}}$. By Bonferroni inequalities, we have 
    \[
    \Pr\left[\bigcup_{\vec{m}\in W}E_{\vec{m}}\right]\leq \sum_{i=1}^{t}(-1)^{i-1}\sum_{X\in \binom{W}{i}}\Pr[E_X].
    \]
    for some odd integer $t=\min\{\frac{\sqrt{a_{\delta} n}}{2},\sqrt{h_{\delta}'\log_q2\cdot \frac{n}{2}}\}$ where $a_{\delta}$ is defined in the proof of \cref{thm:l-vector-bound} and $h_{\delta}'$ is defined in \cref{thm:l-vector-orth-bound}. 

    It is clear that $a:=\Pr[E_{\vec{m}}]=\frac{\Vol_q^S(2n,d-1)}{q^{2n}}=\frac{\sum_{i=0}^{d-1}\binom{n}{i}(q^2-1)^i}{q^{2n}}$. Let $\{\vec{x}_1,\ldots ,\vec{x}_r\}=X\subseteq W$ be any subset of size $r\leq t$, we proceed to bound the probability $E_{X}$. If all vectors in $X$ are linearly independent, we may assume that $\vec{c}_{x_1},\ldots, \vec{c}_{x_h}\in \cC$ and $\vec{c}_{x_{h+1}},\ldots,\vec{c}_{r}\in \cC^{\perp_S}\setminus \cC$ for some $1\leq h\leq r$. For $1\leq \ell\leq h$, we have $\vec{c}_{x_\ell}\in B_{q}^S(2n,d-1)\cap \langle \vec{c}_{x_1}\rangle^{\perp_S}\cap \ldots\cap\langle \vec{c}_{x_{\ell-1}}\rangle^{\perp_S}$, by Colloary~\ref{col:l-vector-bound},
    \[
    \Pr\left[E_{\{\vec{x}_1,\ldots,\vec{x}_\ell\}}\right]=\prod_{i=1}^{\ell}\Pr\left[E_{x_i}\mid E_{\{x_1,\ldots,x_{i-1}\}}\right]\leq \prod_{i=1}^{\ell}(1+q^{-\Omega(n)})\cdot a=(1+q^{-\Omega(n)})^\ell \cdot a^\ell.
    \]
    For $h+1\leq \ell\leq r$, we have $\vec{c}_{x_{\ell}}\in B_{q}^S(2n,d-1)\cap \langle \vec{c}_{x_1}\rangle^{\perp_S}\cap \ldots\cap\langle \vec{c}_{x_{h}}\rangle^{\perp_S}$ and $\vec{c}_{x_{h+1}},\ldots,\vec{c}_{x_\ell}$ are linearly independent. Thus, we apply Corolloary~\ref{col:l-vector-bound} to $\vec{c}_{x_{\ell}}$ and $\vec{c}_{x_1},\ldots, \vec{c}_{x_h}$ to obtain
    \[
    \begin{aligned}
        \Pr\left[E_{\vec{x}_{\ell}}|E_{\{\vec{x}_1,\ldots,\vec{x}_{\ell-1}\}}\right]
        &\leq \left(\frac{(1+q^{-\Omega(n)})}{q^{h}}\Vol(2n,d-1)\right)\cdot \frac{1}{q^{2n-h}-q^{\ell-h-1}} \\
        &\leq (1+q^{-\Omega(n)})\cdot a
    \end{aligned}
    \]
    It follows that
    \[
    \Pr\left[E_X\right]\leq \prod_{i=1}^{r}\Pr\left[E_{\vec{x}_{i}}\mid E_{\{\vec{x}_1,\ldots,\vec{x}_{i-1}\} } \right]\leq \prod_{i=1}^{r}(1+q^{-\Omega(n)})\cdot a=(1+q^{-\Omega(n)})^r\cdot a^r.
    \]
    We now proceed to the case that $\vec{x}_1,\ldots,\vec{x}_r$ are not linearly independent.
    Without loss of generality, we assume that $\vec{x}_1,\ldots,\vec{x}_s$ are the maximally linearly independent vectors in $X$. Then, by the \cref{thm:l-vector-orth-bound}, we have:
    \begin{align*}
    \Pr[E_X] &= \Pr\left[E_{\{\vec{x}_1, \ldots, \vec{x}_s\}}\right] \Pr\left[ E_X | E_{\{\vec{x}_1, \ldots, \vec{x}_s\}}\right] \\
    &\leq (1+q^{-\Omega(n)})^{s}\cdot a^s\cdot \Pr\left[ E_{\vec{x}_r} | E_{\{\vec{x}_1, \ldots, \vec{x}_s\}} \right] \\
    &\leq (1+q^{-\Omega(n)})^{s}\cdot a^s\cdot 2^{-h'_{\delta} n} \\
    &\leq (1+q^{-\Omega(n)})^{s}\cdot a^s\cdot 2^{-h'_{\delta} n}.
    \end{align*}
    Thus, for any linearly independent vectors $Y:=\{\vec{x}_1,\ldots\vec{x}_s\}$, we have the followings
    \begin{align*}
         \sum_{i=s}^t \sum_{\substack{Y \subseteq X \subseteq span\{Y\}, \\ |X| = i}} \Pr[E_X] &\leq (1+q^{-\Omega(n)})^{s}\cdot a^s + (1+q^{-\Omega(n)})^{s}\cdot a^s \sum_{i=s+1}^t \binom{q^s}{i} 2^{-h'_{\delta} n} \\
         &\leq (1+q^{-\Omega(n)})^{s}\cdot a^s\cdot  \left( 1 + \frac{q^{t^2+1}}{t!} 2^{-h'_{\delta} n} \right) \\
         &\leq (1+q^{-\Omega(n)})^{s}\cdot a^s\cdot (1 + 2^{-h'_{\delta}\cdot \frac{n}{2}}),
    \end{align*}
    as $\ell \leq \sqrt{h_{\delta}'\log_q2\cdot \frac{n}{2}}$.
    
    Let $\mathcal{X}_i \subseteq \binom{W}{i}$ be the collection of the set of vectors 
$\{\vec{x}_1, \ldots, \vec{x}_i\}$ that are linearly independent, and define 
$\bar{\mathcal{X}}_i = \binom{W}{i} \setminus \mathcal{X}_i$. Then, we have
\[
\sum_{i=1}^t (-1)^{i-1} \sum_{X \in \binom{W}{i}} \Pr[E_X]
= \sum_{i=1}^t (-1)^{i-1} \sum_{X \in \mathcal{X}_i} \Pr[E_X] 
+ \sum_{i=1}^t (-1)^{i-1} \sum_{X \in \bar{\mathcal{X}}_i} \Pr[E_X]
\]
\[
\leq \sum_{i=1}^{t}(-1)^{i-1}\sum_{X\in \mathcal{X}_i}\Pr[E_X]+\sum_{i=1}^{t}\sum_{X\in \bar{\mathcal{X}_i}}\Pr[E_X]
\leq \sum_{i=1}^t (-1)^{i-1} (1+q^{-\Omega(n)})^{i}\cdot a^i \sum_{X \in \mathcal{X}_i} 
\left( 1 + (-1)^i 2^{-h'_{\delta} n/2} \right).
\]
The size of $\mathcal{X}_i$ is at least 
\[
\frac{1}{i!} \prod_{j=1}^i \left(|W| - q^{j-1}\right).
\]
Thus, we have
\[
1\geq \frac{|\mathcal{X}_i|}{\binom{|W|}{i}} \ge 1 - \frac{\sum_{j=1}^i q^{j-1}}{|W|} \ge 1 - q^{-(2n-k)+t}.
\]
This implies that
\begin{align*}
(-1)^{i-1} a^i \sum_{X \in \mathcal{X}_i} \left( 1 + (-1)^{i-1} 2^{-h'_{\delta}\cdot \frac{n}{2}} \right) 
&=(1+q^{-\Omega(n)})^{i}\cdot a^i|\mathcal{X}_i|\left( 1 + (-1)^{i-1} 2^{-h'_{\delta}\cdot \frac{n}{2}} \right) \\
&\leq (1+q^{-\Omega(n)})^{i}\cdot a^i \binom{|W|}{i} (1 + 2^{-h'_{\delta}\cdot \frac{n}{2}}),
\end{align*}
for odd $i$ and
\begin{align*}
\MoveEqLeft
(-1)^{i-1} (1+q^{-\Omega(n)})^{i}\cdot a^i \sum_{X \in \mathcal{X}_i} \left( 1 + (-1)^{i-1} 2^{-h'_{\delta}\cdot \frac{n}{2}} \right) \\
&\leq - (1+q^{-\Omega(n)})^{i}a^i \binom{|W|}{i} + (1+q^{-\Omega(n)})^{i}\cdot a^i \binom{|W|}{i} (2^{-h'_{\delta}\cdot\frac{n}{2}} + 2^{-\frac{2n-k}{2}}) \\
&\leq - (1+q^{-\Omega(n)})^{i}\cdot a^i \binom{|W|}{i} + (1+q^{-\Omega(n)})^{i}\cdot a^i \binom{|W|}{i} 2^{-h_{\delta}^{*}\cdot \frac{n}{2} + 1}
\end{align*}
\noindent
for even $i$ where $h_{\delta}^{*}=max\{\frac{h'_{\delta}}{2},\frac{2n-k}{2n} \}$. Combining the above results, we obtain:
\begin{align*}
    \Pr\left[\bigcup_{\vec{m}\in W}E_{\vec{m}}\right]&\leq \sum_{i=1}^{t}(-1)^{i-1}(1+q^{-\Omega(n)})^{i}\cdot a^i\binom{|W|}{i}+\sum_{i=1}^{t}(1+q^{-\Omega(n)})^{i}\cdot a^i\binom{|W|}{i}2^{-h_{\delta}^{*}n+1} \\
    &\leq 1.01\left(\sum_{i=1}^{t}(-1)^{i-1}\cdot a^i\binom{|W|}{i}+\sum_{i=1}^{t} a^i\binom{|W|}{i}2^{-h_{\delta}^{*}n+1}\right)
\end{align*}
Define $b=|W|a$ and the right hand side becomes
\begin{align*}
    1.01\left(\sum_{i=1}^{t}(-1)^{i-1}\frac{b^i}{i!}+\sum_{i=1}^{t}\frac{b^i}{i!}2^{-h_{\delta}^{*}n+1}\right)
    &\leq1.01\left(1+\frac{b^{t+1}}{(t+1)!}-e^{-b}+e^b2^{-h_{\delta}^{*}n+1}\right).
\end{align*}
The inequality is due to
\[
e^{b} \geq \sum_{i=0}^{t}\frac{b^{i}}{i!}, \qquad
e^{-b} \geq \sum_{i=0}^{t+1} (-1)^{i}\frac{b^{i}}{i!}.
\]
We concludes the claim by choosing $b=\frac{t+1}{e^2}=\Omega(\sqrt{n})$.
\end{proof}

Combining Lemma~\ref{lem:construct-quantum-codes} and \cref{thm:quantum-GV-Bound}, we obtain a corollary on the existence of quantum codes.

\begin{cor}\label{cor:quantum-codes}
    Let $1\leq k\leq n$, $d=\delta n$ for some $\delta\in (0,1-\frac{1}{q^2})$. There exists a $q$-ary $[[n,n-k,d]]$ quantum code $\cC$ if
    \[
    \frac{q^{2n-k}-1}{q-1}<\frac{c_\delta \sqrt{n}\cdot q^{2n}}{\sum_{i=0}^{d-1}\binom{n}{i}(q^2-1)^i}
    \]
    for some constant $c_{\delta}>0$.
\end{cor}

\bibliographystyle{IEEEtran}
\bibliography{refs}  

\appendix

\section{Basic Form of the Quantum Gilbert-Varshamov Bound}\label{thm:basic-form-GVbound}
\begin{thm}
    Let $d=\delta n$ for some $\delta \in (0,1-\frac{1}{q^2})$. There exists a $q$-ary $[[n,n-k,d]]$ quantum code $\cC$ if 
    \[
    q^{2n-k}-1<\frac{q^{2n}}{\sum_{i=0}^{d-1}\binom{n}{i}(q^2-1)^i}.
    \]
\end{thm}

\begin{proof}
Let $\cC^{\perp_S}$ be a linear symplectic self-orthogonal dual code of length $2n$ and dimension $2n-k$. Let $\vec{c}_m$ be the codeword in $\cC^{\perp_S}$ with the index $\vec{m}\in \F_q^{2n-k}$. As each codeword is symplectically self-orthogonal, it suffices to prove that all random vector $\vec{c}_m$ for $\vec{m}\in \F_q^{2n-k}\setminus \{\vec{0}\}$ has symplectic weight at least $d$. Let $E_{\vec{m}}$ denote the event that a random vector $\vec{c}_m$ has symplectic weight less than $d$. By the union bound, we have
\[
\Pr\left[\bigcup_{\vec{m}\in \F_q^{2n-k}\setminus\{\vec{0}\}}E_{\vec{m}}\right]\leq \sum_{\vec{m}\in \F_q^{2n-k}\setminus \{\vec{0}\}}\Pr[E_{\vec{m}}]=(q^{2n-k}-1)\cdot \frac{\sum_{i=0}^{d-1}\binom{n}{i}(q^2-1)^i}{q^{2n}}.
\]
Consequently, if
\[
(q^{2n-k}-1)\cdot \frac{\sum_{i=0}^{d-1}\binom{n}{i}(q^2-1)^i}{q^{2n}}<1,
\]
there exists $[2n,2n-k,d]$-linear symplectic self-orthogonal dual code $\cC^{\perp_S}$ over $\F_q$. By the Lemma~\ref{lem:construct-quantum-codes}, there exists a $q$-ary $[[n,n-k,d]]$ quantum code $\cC$.
\end{proof}

\section{Proof of Lemma~\ref{lem:quantum-ball}}\label{proof:l-linear-combination}
\begin{proof}
    By \cref{lem:quantum volume}, the number of codeword with weight $r$ is $\binom{n}{r}(q^2-1)^r=2^{\Theta(H_{q^2}(\frac{r}{n}))}$ where the function $H_{q^2}(x)$ is monotone increasing in the domain $x\in (0,1-\frac{1}{q^2})$. Let $\varepsilon>0$, this implies that if $\vec{v}$ is a random vector distributed uniformly at random in $B_q^S(2n,\delta n)$, the probability that $\wt_S(\vec{v})$ is less than $(\delta -\varepsilon)n$ is at most $2^{-a_{\delta}n}$ for some constant $a_\delta\in (0,1)$. We prove our claim by induction. Assume $\ell=2$ and we bound the probability that $\wt_S(\vec{v}_1+\vec{v}_2<\delta n)$. With probability at most $2^{-a_{p}n+1}$, the weight of $\vec{v}_1$ and $\vec{v}_2$ are at most $(\delta-\varepsilon)n$. It follows that
    \begin{equation}
    \label{eq:quantum two vector}
    \begin{aligned}
    \Pr[\wt_S(\vec{v}_1 + \vec{v}_2) < \delta n] 
    &\leq 2 \Pr[\wt_S(\vec{v}_1) < (\delta - \varepsilon)n] \\
    &\quad + \left( \Pr[\wt_S(\vec{v}_i) \geq (\delta - \varepsilon)n,\, i \in [2]] \right) \\
    &\quad \cdot \left( \Pr[\wt_S(\vec{v}_1 + \vec{v}_2) < \delta n 
    \mid \wt_S(\vec{v}_i) \geq (\delta - \varepsilon)n,\, i \in [2]] \right).
    \end{aligned}
    \end{equation}
    Now, we suppose that $\wt_S(\vec{v}_1)=\delta_1n,\ \wt_S(\vec{v}_2)=\delta_2n,\ (\delta-\varepsilon)\leq \delta_1,\delta_2\leq \delta$. Let $\vec{v}_i=(a_{i,1},\ldots,a_{i,n} \mid b_{i,1},\ldots,b_{i,n})$, consider the same probability when each coordinate $(a_{i,j},b_{i,j})$ of $\vec{v}_i$ is chosen to be non-zero element $(0,0)$ with probability $\delta_i/n$ independently. Thus, the expectation of the symplectic weight at the $j$-th coordinate $(a_{i,j},b_{i,j})$ of $\vec{v}_i$ is given by
    \[
    \mathbb{E}[\wt_S(\vec{v}_i)_j]=\delta_i,\ for\ 1\leq j\leq n,
    \]
    where $\wt_S((\vec{v}_i)_j)=1$ if and only if $(a_{i,j},b_{i,j})\neq (0,0)$. \\
    Therefore, the expectation of $\wt_S(\vec{v}_1+\vec{v}_2)_j$ can be computed as:
    \[
    \begin{aligned}
        \E(\wt_S(\vec{v}_1+\vec{v}_2)_j)&=\delta_1(1-\delta_2)+\delta_2(1-\delta_1)+\frac{q^2-2}{q^2-1}\delta_1\delta_2 \\
        &=\delta_1+\delta_2-\frac{q^2}{q^2-1}\delta_1\delta_2 \\
        &\geq 2(\delta-\varepsilon)-\frac{q^2}{q^2-1}(\delta-\varepsilon)^2.
    \end{aligned}
    \]
    As for any $\delta\in (0,1-\frac{1}{q^2})$, $2\delta-\frac{q^2}{q^2-1}\delta^2>\delta$, there exists some constant $\varepsilon>0$ such that
    \[
    \E(\wt_S(\vec{v}_1+\vec{v}_2)_j)\geq 2(\delta-\varepsilon)-\frac{q^2}{q^2-1}(\delta-\varepsilon)^2\geq (\delta+\varepsilon).
    \]
    Since $\mathbb{E}[\wt_S(\vec{v}_1+\vec{v}_2)]=\sum_{j=1}^{n}\mathbb{E}[\wt_S(\vec{v}_1+\vec{v}_2)_j]\geq (\delta+\varepsilon)n$, by \cref{Chernoff}, the probability that $\wt_S(\vec{v}_1+\vec{v}_2)$ is less than $\delta n$ is at most $2^{-b_{\delta}n}$ for some constant $b_\delta\in (0,1)$. We take $c_\delta=min\{a_\delta, b_\delta \}$ and it follows that $\cref{eq:quantum two vector}$ can be proved to be smaller than $2^{-c_{\delta}n}$.
    
    Now we use the induction to prove that the probability of $\wt_S(\sum_{i=1}^{\ell}\vec{v}_i)<\delta n$ is at most $2^{-h_{\delta}'n}$. Assume this holds for $\ell-1$, by induction, we can suppose that the probability of $\wt_S(\sum_{i=1}^{\ell-1}\vec{v}_i)\geq (\delta-\varepsilon)n$ is at least $1-2^{-h_{\delta}n}$. From above argument, we also know that with probability at most $2^{-c_{\delta}n}$, $\wt_S(\vec{v}_\ell)\leq (\delta-\varepsilon)n$. Then, we fix the vector $\vec{v}:=\sum_{i=1}^{\ell-1}\vec{v}_{i}$ and apply the $\ell=2$ argument to $\vec{v}+\vec{v}_\ell$ conditioning that both of the vector has weight at least $(\delta-\varepsilon)n$. Note that the expectation $\mathbb{E}[\wt_S(\vec{v})]\geq (\delta-\varepsilon)n$, following the derivation in the case of $\ell=2$, we obtain the corresponding result that the probability of $\wt_S(\vec{v}+\vec{v}_\ell<\delta n)$ is at most $2^{-c_{\delta}n}$. By union bound, we complete the induction and there exists $h_\delta'\in (0,1)$ such that $\Pr[\wt_S(\sum_{i=1}^{\ell}\vec{v}_i)<\delta n]\leq 2^{-h_{\delta}'n}$. The proof is completed.
\end{proof}

\end{document}